\documentclass[final,hidelinks,onefignum,onetabnum]{siamart251216}

\usepackage{amsfonts}
\usepackage{graphicx}
\usepackage{epstopdf}
\usepackage{algorithmic}
\usepackage{bm}
\usepackage{upgreek}
\usepackage{mathrsfs}
\newcommand{\ii}{\mathrm{i}}

\allowdisplaybreaks

\ifpdf
  \DeclareGraphicsExtensions{.eps,.pdf,.png,.jpg}
\else
  \DeclareGraphicsExtensions{.eps}
\fi

\newsiamremark{remark}{Remark}
\newsiamremark{hypothesis}{Hypothesis}
\crefname{hypothesis}{Hypothesis}{Hypotheses}
\newsiamthm{claim}{Claim}
\newsiamremark{fact}{Fact}
\crefname{fact}{Fact}{Facts}

\headers{Anderson Localization in periodic elastic systems}{X. Feng, P. Li, and W. Wu}

\title{Anderson localization in periodic elastic systems with random perturbations \thanks{Submitted to the editors DATE.
\funding{The work was supported by the National Key R\&D Program of China (2024YFA1012300) and the NSFC grant (12301539).}}}

\author{Xin Feng\thanks{School of Mathematics, Jilin University, Changchun 130000, China (\email{fengxin24@mails.jlu.edu.cn}).}
\and Peijun Li\thanks{SKLMS, Academy of Mathematics and Systems Science, Chinese Academy of Sciences, Beijing 100190, China (\email{lipeijun@lsec.cc.ac.cn}).}
\and Wei Wu\thanks{School of Mathematics, Jilin University, Changchun 130000, China (\email{wei\_wu@jlu.edu.cn}).}}

\usepackage{amsopn}

\newif\ifrevisions
\revisionstrue

\ifpdf
\hypersetup{
  pdftitle={Anderson localization in periodic elastic systems with random perturbations},
  pdfauthor={X. Feng, P. Li, and W. Wu}
}
\fi

\begin{document}

\maketitle

\begin{abstract}
This paper investigates Anderson localization in subwavelength elastic periodic systems with random perturbations. For the unperturbed system, we use layer potential techniques to reformulate the eigenvalue problem as boundary integral equations, and derive asymptotic formulas for the subwavelength eigenvalues. For perturbed systems, we apply the Floquet transform to obtain a periodic formulation and derive equations determining the resonant frequencies under general perturbations. Numerical experiments for perturbed periodic monomers and dimers agree with the theoretical predictions. We further demonstrate Anderson localization by increasing the strength and number of random perturbations. These results provide a mathematical foundation for understanding subwavelength localization in elastic metamaterials.
\end{abstract}



\section{Introduction}

The study of wave localization has grown from a fundamental observation in condensed matter physics into an important theme in mathematical wave theory. In his seminal 1958 paper, Anderson showed that spatial disorder in a lattice potential can suppress electron diffusion, leading to localization \cite{PhysRev.109.1492}. Since then, this concept has been extended from quantum electronic systems to classical wave systems, including acoustic and elastic waves, revealing the universal role of disorder-induced interference. In elastic media, the vector nature of the displacement field makes the spectral analysis more challenging because longitudinal and transverse waves coexist and interact \cite{John1987}. Earlier studies of elastic localization often relied on random matrix theory and multiple scattering in media with disordered elastic parameters \cite{Weaver1990,Sheng2006}. More recently, subwavelength resonators have opened a new regime for studying localization. High-contrast inclusions can generate resonant states that are highly sensitive to aperiodic perturbations. In this setting, layer potential techniques reduce the Lam\'{e} system to a discrete capacitance matrix problem, which captures the hybridization and repulsion of energy levels that are essential for the formation of localized modes \cite{Ammari2024,Ammari2007}.

Motivated by the work \cite{Ammari2024}, this paper investigates localized modes in periodic elastic systems, both in the absence and presence of perturbations. Our analysis is based on layer potential theory \cite{BaoLi2022MaxwellPeriodic, RenChenGaoLi2025SubwavelengthBandgaps, ChenGaoLiRen2025VariationalResonance, AMMARI20175610, ammari2022exceptional}. By representing the solution of the Lam\'{e} system in terms of elastic layer potentials, we transform the eigenvalue problem into a characteristic value problem for boundary integral operators. This formulation, combined with asymptotic analysis, is particularly effective for studying subwavelength resonances. The main contribution of this work is a comprehensive mathematical analysis, together with numerical verification, of localization phenomena in periodic elastic systems both with and without perturbations.

For the unperturbed periodic system, we derive asymptotic descriptions of the subwavelength resonant frequencies. The analysis is divided into two regimes, depending on whether the quasi-momentum is close to zero or bounded away from zero. In the latter case, our result is similar to that in \cite{RenChenGaoLi2025SubwavelengthBandgaps}, but the present approach applies to an arbitrary number of inclusions in each fundamental cell. We then treat the more delicate regime in which the quasi-momentum tends to zero. By analyzing the corresponding Green's function near the simultaneous low-frequency and small-quasi-momentum limit, we overcome the non-self-adjointness of the limiting boundary integral operator and extend the previous results to this singular regime. The study of layer potential operators with small quasi-momentum was initiated for the Helmholtz equation in \cite{ammari2024functional}; however, the elastic system requires a more refined analysis because of the coupling between longitudinal and transverse waves. 

For periodic systems with perturbations, a key difficulty is the lack of a simple Brillouin-zone description. We address this issue by using the Floquet transform, or Bloch--Floquet theory, to connect the perturbed periodic structure with the associated operator formulation \cite{AmmariKosche2024TopologicalHoneycombFloquet}. This reduces the analysis of localized states to an eigenvalue problem for operators analogous to those in \cite[Eq.~(3.7)]{Ammari2024}. The resulting formulation provides a precise characterization of the resonant frequencies and corresponding eigenmodes for periodic systems with general perturbations \cite{AmmariBarandunCaoDaviesHiltunen2024SkinEffect}. Numerical simulations for perturbed periodic monomer and dimer systems agree well with the theoretical predictions and demonstrate the emergence of Anderson localization as the strength and spatial extent of the random perturbations increase.

The paper is organized as follows. Section 2 introduces the governing equations, the geometry of the elastic system, and the layer potential operators, and then develops the corresponding layer potential formulation. Section 3 derives asymptotic results for the resonant frequencies of the periodic elastic system, treating separately the cases where the quasi-momentum is away from zero and close to zero. Section 4 presents the main results on the existence of localized eigenmodes in randomly perturbed periodic systems, together with numerical simulations.

\section{Problem formulation}

We study the propagation of elastic waves in a three-dimensional medium containing periodically arranged elastic scatterers. This section introduces the geometry of the system and the governing equations for elastic waves.

\subsection{Model equations}

Let $\Lambda$ be a lattice of dimension $d_l\in\{1,2,3\}$ generated by the basis vectors ${\bm{l}_1,\ldots,\bm{l}_{d_l}}$, i.e.,
$$
\Lambda:= \left\{m_1\bm{l}_1 + \cdots + m_{d_l}\bm{l}_{d_l} : m_i\in\mathbb{Z}, i=1, \cdots, d_l\right\}.
$$
The corresponding fundamental cell is defined by
$$
  Y:=\left\{c_1\bm{l}_1+\cdots+c_{d_l}\bm{l}_{d_l}: -\frac{1}{2}\leq c_i\leq \frac{1}{2},\, i=1,\ldots,d_l\right\}\times\mathbb{R}^{3-d_l}.
$$
For $\bm{x}\in\mathbb{R}^3$, we write $\bm{x}=(\bm{x}_l,\bm{x}_0)$, where $\bm{x}_l:=(x_1,\ldots, x_{d_l})$ and $\bm{x}_0:=(x_{d_l+1},\ldots, x_3)$. Assume that the fundamental cell $Y$ contains $N$ bounded and connected scatterers $D_i$, $i=1,\ldots,N$ with $C^2$ boundary, and their closures and all lattice translates are mutually separated. For each $\bm{m}\in\Lambda$, let $D_i^{\bm{m}} := D_i+\bm{m}, \bm{m}\in\Lambda$ denote the translated copy of the $i$-th scatterer. We define
$$
D_{\rm tot}^{\bm{m}} := \cup_{i\in\{1,\ldots, N\}} D_i^{\bm{m}}, \quad D_{\rm tot}:=\cup_{\bm{m}\in\Lambda} D_{\rm tot}^{\bm{m}}, \quad D:=\cup_{i\in\{1,\ldots,N\}} D_i.
$$
Thus, $D_{\rm tot}$ denotes the union of all scatterers in the entire periodic system, while $D$ denotes the union of the scatterers within the fundamental cell.

The background medium $\mathbb{R}^3\setminus\overline{D}_{\rm tot}$ is assumed to be homogeneous, with constant Lam\'{e} parameters $\lambda, \mu$ and density $\rho$. Each scatterer $D_i^{\bm m}$, $\bm m\in\Lambda$, is characterized by Lam\'{e} parameters $\tilde{\lambda}_i, \tilde{\mu}_i$ and density $\tilde{\rho}_i^{\bm m}$. We assume that the Lam\'{e} parameters are the same in all periodic cells, whereas the density may vary from cell to cell, thereby introducing perturbations into the periodic system. More precisely, the material parameters inside and outside the scatterers satisfy
\begin{equation}\label{eq:parameter}
\tilde{\lambda}_i = \lambda/\delta_i, \quad \tilde{\mu}_i = \mu/\delta_i, \quad \tilde{\rho}_i^{\bm{m}} =\rho/\epsilon_i^{\bm{m}},
\end{equation}
where $\delta_i$ and $\epsilon_i^{\bm m}$ characterize the material contrast between the background medium and the scatterers. The parameter $\delta_i$ is fixed throughout the lattice, while $\epsilon_i^{\bm m}$ depends on the lattice vector $\bm m$ and thus models the perturbation of the system.

The shear and compressional wave speeds in the background medium are given by $c_s:=(\mu/\rho)^{1/2}$ and $c_p:=((\lambda+2\mu)/\rho)^{1/2}$, respectively. The corresponding wave speeds inside $D_i^{\bm{m}}$ are $\tilde{c}_{s,i}^{\bm{m}} := (\tilde{\mu}_i/\tilde{\rho}_i^{\bm{m}})^{1/2}$ and $\tilde{c}_{p,i}^{\bm{m}} := ((\tilde{\lambda}_i + 2\tilde{\mu}_i)/\tilde{\rho}_i^{\bm{m}})^{1/2}$. It follows from \eqref{eq:parameter} that
$$
\tau_i^{\bm{m}} := c_s/\tilde{c}_{s,i}^{\bm{m}} = c_p/\tilde{c}_{p,i}^{\bm{m}} = (\delta_i/\epsilon_i^{\bm{m}})^{1/2}.
$$
Throughout this paper, we assume that $\mu>0, 3\lambda+2\mu>0$ and $\rho>0$. The contrasts satisfy $\delta_i>0$ and $\epsilon_i^{\bm m}>0$, with $\delta_i\ll 1$, and that $0<\tau_-\leq\tau_i^{\bm{m}}\leq\tau_+<\infty, i=1,\ldots, N, \bm{m}\in\Lambda$, where $\tau_-, \tau_+$ are fixed constants.

Since
$$
\mathcal{L}^{\tilde{\lambda}_i, \tilde{\mu}_i} + \tilde{\rho}_i^{\bm{m}}\omega^2 = \delta_i^{-1}\mathcal{L}^{\lambda,\mu} + (\epsilon_i^{\bm{m}})^{-1}\rho\omega^2,
$$
the elastic wave propagation in the entire system can be formulated as
\begin{equation}\label{eq:elas_sys}
	\left\{
	\begin{array}{ll}
	\displaystyle (\mathcal{L}^{\lambda, \mu} +\rho\omega^2) \bm{u} = \bm{0}  &   \mathrm{in }\  \mathbb{R}^3\setminus\overline{D}_{\rm tot}, \medskip \\
	\displaystyle (\mathcal{L}^{\lambda, \mu} +\rho(\tau_i^{\bm{m}})^2\omega^2) \bm{u} = \bm{0}  &  \mathrm{in }\  D_i^{\bm{m}},\quad i=1,\ldots, N, \ \bm{m}\in\Lambda, \medskip\\
	\displaystyle \bm{u}|_+ - \bm{u}|_- = \bm{0}   & \mathrm{on }\  \partial D_{\rm tot}, \medskip \\
    \displaystyle \delta_i \partial_{\bm\nu} \bm{u}|_+ - \partial_{\bm\nu}\bm{u}|_-  = \bm{0}   &\mathrm{on }\  \partial D_i^{\bm{m}}, \quad i=1, \cdots, N,\ \bm{m}\in\Lambda,
	\end{array}
	\right.
\end{equation}
together with the outgoing radiation condition for $\bm{u}(\bm{x}_l, \bm{x}_0)$ as $|\bm{x}_0|\rightarrow\infty$. When $d_l=3$, the system is periodic in all spatial directions, and no radiation condition is imposed on $\bm u$. Here, $\mathcal{L}^{\lambda,\mu}$ denotes the Lam\'{e} operator associated with the Lam\'{e} constants $\lambda$ and $\mu$, defined by
\begin{equation*}
\mathcal{L}^{\lambda, \mu}\bm{u} := \mu\Delta \bm{u} +(\lambda+ \mu)\nabla\nabla\cdot \bm{u}.
\end{equation*}
The corresponding conormal derivative on the boundary $\partial\Omega$ of a bounded domain $\Omega$ is given by
\begin{equation*}
\partial_{\bm\nu} \bm{u}:=\lambda(\nabla\cdot \bm{u})\bm{\nu} +\mu(\nabla\bm{u}+\nabla\bm{u}^{\top} )\bm{\nu},
\end{equation*}
where $\bm\nu$ is the unit outward normal to $\partial\Omega$, and the superscript $\top$ denotes matrix transpose.

Let $\Lambda^*$ denote the dual lattice of $\Lambda$, generated by the vectors $\bm{\beta}_1, \ldots, \bm{\beta}_{d_l}$ satisfying $\bm{\beta}_i\cdot \bm{l}_j = 2\pi\delta_{ij}$ for $i,j=1, \ldots, d_l$. The Brillouin zone associated with $\Lambda$ is defined by $Y^*:=(\mathbb{R}^{d_l}\times\{\bm{0}\})/\Lambda^*$.

\subsection{Quasi-periodic layer potentials}

The main tools used in this paper are layer potential operators for the Lam\'{e} system, both for a single inclusion and for periodically arranged inclusions. These operators have been extensively studied; see, for example, \cite{ammari2015mathematical, AmmariKangLee2009}. Therefore, in this section we recall only their definitions and basic properties, without providing detailed proofs.

Let $\bm{G}^{\omega}$ be the fundamental solution of the operator $\mathcal{L}^{\lambda,\mu}+\rho\omega^2$, satisfying
$$
 (\mathcal{L}^{\lambda,\mu}+\rho\omega^2)\bm{G}^{\omega}(\bm{x},\bm{y}) = \delta(\bm{x}-\bm{y})\bm{I}_3,
$$
where $\bm{I}_3$ is the $3\times 3$ identity matrix. We define the single- and double-layer potentials of a bounded domain $\Omega$ with $C^2$ boundary by
\begin{equation}\label{eq:singlelayer}
\begin{aligned}
    \mathcal{S}_\Omega^{\omega}[\bm{\phi}](\bm{x}) &= \int_{\partial\Omega} \bm{G}^{\omega}(\bm{x},\bm{y})\bm{\phi}(\bm{y})\mathrm{d}\sigma(\bm{y}), \quad \bm{x}\in\mathbb{R}^3,\\
    \mathcal{D}_\Omega^{\omega}[\bm{\phi}](\bm{x}) &= \int_{\partial\Omega} \partial_{\bm\nu(\bm y)}\bm{G}^{\omega}(\bm{x},\bm{y})\bm{\phi}(\bm{y})\mathrm{d}\sigma(\bm{y}), \quad \bm{x}\in\mathbb{R}^3\backslash\partial\Omega.
\end{aligned}
\end{equation}

It is known from \cite[Section~2.15.2]{Ammari2018PhotonicsPhononics} that the single-layer potential and the conormal derivative of the double-layer potential are continuous across $\partial D$. In contrast, the jump relations are satisfied by the double-layer potential and by the conormal derivative of the single-layer potential as follows:
\begin{equation}\label{eq:jumprelation}
\begin{aligned}
   \partial_{\bm\nu}\mathcal{S}_D^{\omega}[\bm{\phi}]|_{\pm}(\bm{x})& = \big(\pm\frac{1}{2}\mathcal{I} + (\mathcal{K}_D^{\omega})^*\big)[\bm{\phi}](\bm{x}),\quad \bm{x}\in\partial D, \\
    (\mathcal{D}_D^{\omega}[\bm{\phi}])|_{\pm}(\bm{x}) & = \big(\mp\frac{1}{2}\mathcal{I}+\mathcal{K}_D^{\omega}\big)[\bm{\phi}](\bm{x}),\quad \bm{x}\in\partial D,
\end{aligned}
\end{equation}
where $\mathcal{I}$ is the identity operator, $\mathcal{K}_D^{\omega}$ is defined by
\begin{equation}\label{eq:singlek}
    \mathcal{K}_D^{\omega}[\bm{\phi}](\bm{x}) = \mathrm{p.v.}\int_{\partial D} \partial_{\bm\nu({\bm y})} \bm{G}^{\omega}(\bm{x},\bm{y})\bm{\phi}(\bm{y})\mathrm{d}\sigma(\bm{y}),\quad \bm{x}\in\partial D,
\end{equation}
and $(\mathcal{K}_D^{\omega})^*$ given by
\begin{equation}\label{eq:singlekstar}
    (\mathcal{K}_D^{\omega})^*[\bm{\phi}](\bm{x}) = \mathrm{p.v.}\int_{\partial D} \partial_{\bm\nu({\bm x})} \bm{G}^{\omega}(\bm{x},\bm{y})\bm{\phi}(\bm{y})\mathrm{d}\sigma(\bm{y}),\quad \bm{x}\in\partial D.
\end{equation}

It was shown in \cite{VodickaMantic2004} that, for an open connected set $\Omega\subset\mathbb{R}^2$ with Lipschitz boundary, there exist at most two scaling factors $\rho_1, \rho_2$ such that the single-layer potential defined on $\rho_i\Omega:=\{\rho_i\bm{x}\in\mathbb{R}^2 : \bm{x}\in\Omega\}$ is not invertible. Thus, by choosing the size of each scatterer appropriately, we may ensure that the required invertibility condition holds. Throughout this paper, we assume that the scatterers $D_i$ are chosen so that $\mathcal{S}_{D_i}^0$ is invertible for each $i$.

For $\bm{\alpha}\in Y^*$, we introduce the quasi-periodic Green's function $\bm{G}^{\bm{\alpha},\omega}$ satisfying
\begin{equation*}
(\mathcal{L}^{\lambda, \mu} +\rho\omega^2)\bm{G}^{\bm{\alpha},\omega}(\bm{x},\bm{y})=\sum_{\bm{n} \in \Lambda} \delta(\bm{x}-\bm{y}-\bm{n})  e^{{\rm{i}} \bm{n}\cdot \bm{\alpha}}\bm{I}_3.
\end{equation*}
The corresponding quasi-periodic single-layer potential $\mathcal{S}_{\Omega}^{\bm\alpha,\omega}$, double-layer potential $\mathcal{D}_{\Omega}^{\bm\alpha,\omega}$, and boundary integral operators $\mathcal{K}_{\Omega}^{\bm\alpha,\omega}$ and $(\mathcal{K}_{\Omega}^{\bm\alpha,\omega})^*$ are defined analogously to \eqref{eq:singlelayer}, \eqref{eq:singlek}, and \eqref{eq:singlekstar}, with $\bm G^\omega$ replaced by $\bm G^{\bm\alpha,\omega}$. They satisfy the same jump relations as those in \eqref{eq:jumprelation}. When restricted to $\partial\Omega$, the single-layer operators $\mathcal{S}_\Omega^{\omega}$ and $\mathcal{S}_\Omega^{\bm\alpha,\omega}$ are bounded from $(H^{-1/2}(\partial \Omega))^3$ to $(H^{1/2}(\partial \Omega))^3$. Moreover, the operators $\mathcal{K}_\Omega^{\omega}$ and $\mathcal{K}_\Omega^{\bm\alpha,\omega}$ are bounded on $(H^{-1/2}(\partial\Omega))^3$. It is worth noting that, by definition, $\bm{G}^{\bm{\alpha},\omega}(\bm{x},\bm{y})$ depends only on the difference $\bm{x}-\bm{y}$. Therefore, in the subsequent discussion, we write it as $\bm{G}^{\bm{\alpha},\omega}(\bm{x}-\bm{y})$, or simply as $\bm{G}^{\bm{\alpha},\omega}(\bm{x})$ when no confusion can arise.

\section{Periodic systems}

Before analyzing the perturbed periodic system \eqref{eq:elas_sys}, we first consider the corresponding unperturbed case. Specifically, we assume that $\epsilon_i^{\bm m}:=\epsilon_i$, so that the material parameters are periodic throughout the lattice. We seek a quasi-periodic solution of \eqref{eq:elas_sys}, i.e., for some $\bm\alpha\in Y^*$, the displacement field $\bm u$ satisfies
\begin{equation}\label{eq:quasicondition}
    \bm{u}(\bm{x}+\bm{m}) = e^{\ii\bm{\alpha}\cdot \bm{m}}\bm{u}(\bm{x}), \quad\forall\, \bm{m}\in\Lambda.
\end{equation}
The vector $\bm\alpha$ is called the quasi-periodicity of $\bm u$.

The discussion in this section consists of three parts. First, we reformulate \eqref{eq:elas_sys} as an operator equation and recall the necessary properties of the relevant layer potential operators. We then study the number of resonant modes and derive asymptotic approximations for the resonant frequencies in the regime $|\bm{\alpha}|^2\geq\rho\omega^2/(\lambda+2\mu)$. This analysis follows the general approach of \cite{RenChenGaoLi2025SubwavelengthBandgaps} and \cite{ammari2024functional}. The small quasi-momentum limits for $d_l=1,2$ are treated in Section \ref{smallalpha}. 

\subsection{Layer potential operators}

By Bloch's theorem, in the periodic setting of \eqref{eq:elas_sys}, we seek an $\bm\alpha$-quasi-periodic solution $\bm u$ satisfying \eqref{eq:quasicondition}. We represent $\bm u$ in terms of layer potentials as
\begin{equation}\label{eq:int_repre}
\bm{u}(\bm{x})=\left\{
	\begin{array}{ll}
		\mathcal{S}^{\bm{\alpha},\omega}_D[\bm{\phi}](\bm{x}),  &  \bm{x} \in  Y\setminus\overline{D}, \medskip \\
                 \displaystyle  \mathcal{S}_{D_i}^{\tau_i\omega}[\bm{\psi}](\bm{x}) ,   &\bm{x} \in  D_i.
	\end{array}
	\right.
\end{equation}

Define the operator $\mathcal{S}_D^{\bm{\tau}\omega}: (H^{-1/2}(\partial D))^3\rightarrow (H^{1/2}(\partial D))^3$ by $\mathcal{S}_D^{\bm{\tau}\omega} := \mathcal{S}_{D_i}^{\tau_i\omega}$ on $\partial D_i$, and similarly define $\mathcal{K}_D^{\bm{\tau}\omega}: (H^{-1/2}(\partial D))^3\rightarrow (H^{-1/2}(\partial D))^3$ by $\mathcal{K}_D^{\bm{\tau}\omega}:=\mathcal{K}_{D_i}^{\tau_i\omega}$ on $\partial D_i, i=1,\ldots, N$. Let $\bm{\delta}:=\mathrm{diag}(\delta_1, \cdots, \delta_N)$. Using the jump relations in \eqref{eq:jumprelation} together with the representation formula \eqref{eq:int_repre}, we find that solving \eqref{eq:elas_sys} is equivalent to finding a nontrivial pair of boundary densities $(\bm\psi,\bm\phi)$ such that
\begin{equation*}
\mathcal{A}_{\bm{\delta}}^{\bm{\alpha},\omega}\begin{pmatrix}
     \bm{\psi} \\
     \bm{\phi}
\end{pmatrix}
=\bm{0},
\end{equation*}
where
\begin{equation*}
\mathcal{A}_{\bm{\delta}}^{\bm{\alpha},\omega}:=\left(
\begin{array}{cc}
\mathcal{S}^{\bm{\tau}\omega}_D & -\mathcal{S}^{\bm{\alpha},\omega}_D\\
-\frac{1}{2}\mathcal{I}+(\mathcal{K}_D^{\bm{\tau}\omega})^{\ast} &-\big(\frac{1}{2}\mathcal{I}+(\mathcal{K}_D^{\bm{\alpha},\omega})^{\ast} \big)\bm{\delta}
\end{array}
\right).
\end{equation*}

\begin{theorem}(\cite[Section 2.15.1]{Ammari2018PhotonicsPhononics})\label{thm:greenexpansionsingle}
The Green's function $\bm{G}^{\omega}=(G_{ij}^{\omega})_{i,j=1}^3$ admits the following expansion with respect to the frequency $\omega$:
\begin{align*}
    G_{ij}^{\omega}(\bm{x})& =  -\frac{1}{4\pi\rho}\sum\limits_{n=0}^{+\infty}\frac{\ii^n}{(n+2)n!}\left(\frac{n+1}{c_s^{n+2}} + \frac{1}{c_p^{n+2}}\right)\omega^n\delta_{ij}|\bm{x}|^{n-1} \\
    &\quad  + \frac{1}{4\pi\rho}\sum\limits_{n=0}^{+\infty}\frac{\ii^n(n-1)}{(n+2)n!}\left( \frac{1}{c_s^{n+2}}-\frac{1}{c_p^{n+2}}\right)\omega^n|\bm{x}|^{n-3}x_ix_j.
\end{align*}
\end{theorem}

Since the coefficient of the first-order term in $\omega$ in the expansion of $G_{ij}^{\omega}$ is independent of the spatial variable, we have
\begin{equation}\label{eq:potentialsingle}
\mathcal{S}_{D,1}^0[\bm{\phi}] = -\frac{\ii}{12\pi\rho}\left(\frac{2}{c_s^3}+\frac{1}{c_p^3}\right)\int_{\partial D}\bm{\phi}(\bm{y})\mathrm{d}\sigma(\bm{y}), \quad \mathcal{K}_{D,1}^0[\bm{\phi}] = 0.
\end{equation}
Using Theorem~\ref{thm:greenexpansionsingle}, we obtain the following expansions for the single-layer operator and the Neumann--Poincar\'{e} operator:
\begin{equation}\label{eq:expansionsingle}
    \mathcal{S}_D^{\omega} = \mathcal{S}_D^0 + \sum\limits_{j=1}^{+\infty} \omega^{j}\mathcal{S}_{D,j}^0, \quad  \mathcal{K}_D^{\omega} = \mathcal{K}_D^0 + \sum\limits_{j=2}^{+\infty} \omega^{j}\mathcal{K}_{D,j}^0.
\end{equation}

\begin{theorem}\label{thm:greenexpansion}
For non-zero quasi-momentum $\bm\alpha$, the quasi-periodic Green's function ${\bm G}^{\bm\alpha,\omega}$ admits an expansion in even powers of $\omega$ of the form
$$
    \bm{G}^{\bm{\alpha},\omega} = \bm{G}^{\bm{\alpha},0} + \sum\limits_{j=1}^{+\infty} \omega^{2j}\bm{G}_{j}^{\bm{\alpha},0},
$$
where the coefficients $\bm G_{j}^{\bm\alpha,0}$, $j=0,1,2,\ldots$, are independent of $\omega$.
\end{theorem}

\begin{proof}
It follows from the representation formulas (SM2.1)--(SM2.2), (SM2.4)--(SM2.6), (SM2.8)--(SM2.10) in Section~SM2 of the supplementary material that the dependence of ${\bm G}^{\bm\alpha,\omega}$ on $\omega$ occurs only through $\omega^2$, i.e., $\bm{G}^{\bm{\alpha},\omega} = \bm{G}^{\bm{\alpha},-\omega}$. Thus ${\bm G}^{\bm\alpha,\omega}$ is an even function of $\omega$. Consequently, its expansion contains only even powers of $\omega$, and all odd order terms vanish.
\end{proof}

By Theorem~\ref{thm:greenexpansion}, for fixed non-zero quasi-momentum, the operators $\mathcal{S}_D^{\bm\alpha,\omega}$ and $\mathcal{K}_D^{\bm\alpha,\omega}$ admit the expansions
\begin{equation*}
    \mathcal{S}_D^{\bm{\alpha},\omega} = \mathcal{S}_D^{\bm{\alpha},0} + \sum\limits_{j=1}^{+\infty} \omega^{2j}\mathcal{S}_{D,j}^{\bm{\alpha},0}, \quad  \mathcal{K}_D^{\bm{\alpha},\omega} = \mathcal{K}_D^{\bm{\alpha},0} + \sum\limits_{j=1}^{+\infty} \omega^{2j}\mathcal{K}_{D,j}^{\bm{\alpha},0}.
\end{equation*}
Following the argument in \cite[Lemma 2.3]{RenChenGaoLi2025SubwavelengthBandgaps}, one can show that, for $\bm\alpha\neq 0$, the operator $\mathcal{S}_D^{\bm\alpha,0}$ is invertible, with inverse $(\mathcal{S}_D^{\bm{\alpha},0})^{-1}:(H^{1/2}(\partial D))^3\rightarrow (H^{-1/2}(\partial D))^3$.

Let $\Psi_{\mathbb{R}^3}$ denote the space of rigid motions in $\mathbb{R}^3$, defined by
$$
    \Psi_{\mathbb{R}^3} := \left\{\bm{a}+\bm{B}\bm{x}: \bm{a}\in\mathbb{R}^3,\, \bm{B}\in M_3^A\right\},
$$
where $M_3^A$ denotes the space of antisymmetric $3\times 3$ matrices. The space $\Psi_{\mathbb{R}^3}$ is six-dimensional, and a basis is given by
$$
    \bm{\psi}_1:= \begin{pmatrix} 1 \\ 0 \\ 0\end{pmatrix},  \bm{\psi}_2:=\begin{pmatrix} 0 \\ 1 \\ 0\end{pmatrix}, \bm{\psi}_3:=\begin{pmatrix} 0 \\ 0 \\ 1\end{pmatrix}, \bm{\psi}_4:= \left(\!\begin{matrix} -x_2 \\ x_1 \\ 0\end{matrix}\!\right),  \bm{\psi}_5:= \left(\!\begin{matrix} 0 \\ -x_3 \\ x_2 \end{matrix}\!\right),   \bm{\psi}_6:= \left(\!\begin{matrix} x_3 \\ 0 \\ -x_1 \end{matrix}\!\right).
$$
Moreover, $\Psi_{\mathbb{R}^3}$ consists precisely of the solutions to the homogeneous Neumann problem
\begin{equation}\label{eq:poissonproblem}
\left\{
\begin{aligned}
    \mathcal{L}^{\lambda,\mu}\bm{u} &= 0, && \bm{x}\in D_i, \\
   \partial_{\bm\nu} \bm{u}&=0, && \bm{x}\in\partial D_i,
\end{aligned}
\right.
\end{equation}
for each $i=1,\ldots,N$.

To distinguish the boundary traces on different inclusions, we define $\bm\psi^{(i,j)}:=\bm\psi_j|_{\partial D_i}, i=1,\ldots,N, j=1,\ldots,6$. For each $i=1,\ldots,N$ and $j=1,\ldots,6$, let $\bm\psi_i^{(j)}\in (H^{-1/2}(\partial D))^3$ be given by
\[
\bm\psi_i^{(j)}
=
\begin{cases}
\bm\psi^{(i,j)} & \text{on } \partial D_i,\\
\bm 0 & \text{on } \partial D_k,\quad k\neq i.
\end{cases}
\]
Then the null space of $-\frac{1}{2}\mathcal{I}+\mathcal{K}_D^0$ on $(H^{-1/2}(\partial D))^3$ is
$$
    \Psi := \mathrm{Span}\left\{\bm{\psi}^{(j)}_i : i=1,\ldots, N,\, j=1, \ldots, 6\right\}.
$$

The following lemma is a direct consequence of \cite[Lemma~1.3]{ammari2015mathematical}.

\begin{lemma}
The space $\Psi$ is the null space of both $-\frac{1}{2}\mathcal{I}+\mathcal{K}_D^0$ and $-\frac{1}{2}\mathcal{I}+\mathcal{K}_D^{\bm{\alpha},0}$ on $(H^{-1/2}(\partial D))^3$.
\end{lemma}

We now characterize the kernel of $-\frac{1}{2}\mathcal{I}+(\mathcal{K}_D^{\bm{\alpha},0})^*$. By the invertibility of $\mathcal{S}_D^{\bm{\alpha},0}$, define
\begin{equation*}
{\bm f}_{i,\mathrm{per}}^{(j)}:=(\mathcal{S}_D^{\bm\alpha,0})^{-1}[\bm\psi_i^{(j)}]
\in (H^{-1/2}(\partial D))^3, \quad i=1,\ldots,N,\, j=1,\ldots,6.
\end{equation*}
Since $\bm{\psi}_i^{(j)}$ satisfies \eqref{eq:poissonproblem} in $D_i$, the jump relations \eqref{eq:jumprelation} imply that $\bm{f}_{i,\mathrm{per}}^{(j)}\in\mathrm{Ker}(-\frac{1}{2}\mathcal{I}+(\mathcal{K}_D^{\bm{\alpha},0})^*)$. Conversely, for any $\bm{u}\in \mathrm{Ker}(-\frac{1}{2}\mathcal{I}+(\mathcal{K}_D^{\bm{\alpha},0})^*)$, the single-layer potential $\mathcal{S}_D^{\bm\alpha,0}[\bm u]$ is $\bm\alpha$-quasi-periodic and, when restricted to $D$, solves the homogeneous Neumann problem \eqref{eq:poissonproblem}. Hence its boundary trace belongs to $\Psi$.

Similarly, one can show that the functions $\bm{f}_i^{(j)}:=(\mathcal{S}_D^0)^{-1}[\bm{\psi}_i^{(j)}]\in (H^{-1/2}(\partial D))^3$, $i=1,\ldots, N, j=1,\ldots,6$ span the kernel of $-\frac{1}{2}\mathcal{I}+(\mathcal{K}_D^0)^*$. We summarize these characterizations in the following lemma.

\begin{lemma}
The kernels of $-\frac{1}{2}\mathcal{I}+(\mathcal{K}_D^{\bm\alpha,0})^*$ and $-\frac{1}{2}\mathcal{I}+(\mathcal{K}_D^0)^*$ are given by
\begin{align*}
    \mathrm{Ker}(-\frac{1}{2}\mathcal{I}+(\mathcal{K}_D^{\bm{\alpha},0})^*)& = \mathrm{Span}\left\{\bm{f}_{i,\mathrm{per}}^{(j)}: i=1,\ldots, N,\, j=1,\ldots, 6\right\}, \\
    \mathrm{Ker}(-\frac{1}{2}\mathcal{I}+(\mathcal{K}_D^0)^*)& = \mathrm{Span}\left\{\bm{f}_i^{(j)}: i=1,\ldots, N,\, j=1,\ldots, 6\right\}.
\end{align*}
\end{lemma}

\subsection{Large quasi-momentum regime}

We first consider the case where all the modes are evanescent. To ensure this, $\omega$ must satisfy $|\omega|<c_s\min_{\bm l\in\Lambda^*}|\bm\alpha+\bm l|$. Since the lemmas and theorems in this subsection are mostly direct generalizations of existing results in \cite{LI2026113822, RenChenGaoLi2025SubwavelengthBandgaps, ammari2024functional}, we present only their statements and omit the detailed proofs.

We begin with a relation between $\mathcal{S}_D^0$ and $(\mathcal{K}_{D,2}^0)^*$. The following lemma generalizes \cite[(4.13)]{LI2026113822}.

\begin{lemma}
For any $\bm{\phi}\in (L^2(\partial D))^3$, the following identity holds:
    $$
        -\rho\tau_i^2\int_{D} \mathcal{S}_{D}^0[\bm{\phi}]\cdot\bm{\psi}_i^{(j)}\mathrm{d}x = \int_{\partial D} (\mathcal{K}_{D,2}^0)^*[\bm{\phi}]\cdot\bm{\psi}_i^{(j)}\mathrm{d}\sigma.
    $$
\end{lemma}

We refer to values of $\omega$ for which $\mathcal{A}_{\bm{\delta}}^{\bm{\alpha},\omega}$ is not invertible as Bloch resonant frequencies. The following lemma gives the number of Bloch resonant frequencies associated with $\mathcal{A}_{\bm{\delta}}^{\bm{\alpha},\omega}$. It is a direct generalization of \cite[Lemma~3.5]{RenChenGaoLi2025SubwavelengthBandgaps}.

\begin{lemma}
Let $\bm{\alpha}\in Y^*$ be fixed, and assume that $|\bm{\alpha}|^2\geq\rho\omega^2/(\lambda+2\mu)$. Then the operator $\mathcal{A}_{\bm{\delta}}^{\bm{\alpha},\omega}$ has precisely $12N$ Bloch resonant frequencies $\omega_i^{\bm{\alpha},(j)} = \omega_i^{\bm{\alpha}, (j)}(\bm{\delta})$, $i=1, \ldots, N$, $j=1, \ldots, 6$, counted with multiplicity, such that $\omega_i^{\bm{\alpha}, (j)}(0) = 0$, and each $\omega_i^{\bm\alpha,(j)}$ depends continuously on $\bm\delta$. Among these resonant frequencies, precisely $6N$ have positive real part and correspond to physical resonances.
\end{lemma}

The following theorem on resonant frequencies, together with Proposition~\ref{prop:normalalpha-field}, can be proved by following the argument of \cite[Theorem~3.12]{RenChenGaoLi2025SubwavelengthBandgaps}.

\begin{theorem}\label{thm:resonantfreq}
Let $\bm{\alpha}\in Y^*$ and assume that $|\bm{\alpha}|^2\geq\rho\omega^2/(\lambda+2\mu)$. For $i,i'=1,\ldots,N$ and $j,j'=1,\ldots,6$, define
\[
\displaystyle\widehat{\bm{C}}_{i,i'}^{(j,j')}:= -\int_{\partial D} \bm{f}_{i,\mathrm{per}}^{(j)}\cdot\bm{\psi}_{i'}^{(j')}\mathrm{d}\sigma, \quad \displaystyle\bm{M}_{i,i'}^{(j,j')}:=\int_{D}\bm{\psi}_i^{(j)}\cdot\bm{\psi}_{i'}^{(j')}\mathrm{d}\bm{x}.
\]
Let $\mathfrak{M}\in\mathbb{R}^{6N\times 6N}$ and $\widehat{\mathfrak{C}}\in\mathbb{C}^{6N\times 6N}$ be the matrices whose entries are given by
\[
 \mathfrak{M}_{6(i-1)+j, 6(i'-1)+j'} = \bm{M}_{i,i'}^{(j,j')},\quad \widehat{\mathfrak{C}}_{6(i-1)+j, 6(i'-1)+j'} = \widehat{\bm{C}}_{i,i'}^{(j,j')}.
\]
Then both $\mathfrak M$ and $\widehat{\mathfrak C}$ are positive definite.

Define the diagonal matrices $\uptau:=\mathrm{diag}(\tau_1^2,\cdots, \tau_N^2)\in\mathbb{R}^{6N\times 6N}$ and  $\updelta:=\mathrm{diag}(\delta_1, \allowbreak \cdots, \delta_N)\in\mathbb{R}^{6N\times 6N}$, where each $\tau_i^2$ and $\delta_i$ is repeated six consecutive times. Let $\lambda_k^{\bm\alpha}$, $k=1,\ldots,6N$, be the eigenvalues of $\widehat{\mathfrak{C}}\updelta(\mathfrak{M}\uptau)^{-1}$. Then the subwavelength resonant frequencies $\omega_k^{\bm{\alpha}} = \omega_k^{\bm{\alpha}}(\bm{\delta})$ of $\mathcal{A}_{\bm{\delta}}^{\bm{\alpha},\omega}$ satisfy
$$
\omega_k^{\bm{\alpha}} := (\lambda_k^{\bm{\alpha}}/\rho)^{1/2} + \mathcal{O}(\omega^2) + \mathcal{O}(\bm{\delta}) = \mathcal{O}({\bm\delta}^{1/2}),\quad k=1, \ldots, 6N.
$$
\end{theorem}

Hereafter, for a scalar function $g$, the notation $\mathcal{O}(g(\bm\delta))$ denotes $\mathcal{O}\!\bigl(\!\max_{1\leq i\leq N} \! g(\delta_i)\!\bigr)$.

\begin{proposition}\label{prop:normalalpha-field}
As $\delta_i\to 0$, $i=1,\ldots,N$, the solution $\bm u$ of the unperturbed periodic problem \eqref{eq:elas_sys}, with $\bm\alpha\in Y^*$ satisfying $|\bm{\alpha}|^2\geq\rho\omega^2/(\lambda+2\mu)$ and $\omega=\mathcal{O}({\bm\delta}^{1/2})$, has the following asymptotic form inside each inclusion:
    \begin{equation*}
        \bm{u}|_{D_i} = \sum\limits_{j=1}^6 c_i^{(j)}\bm{\psi}_i^{(j)} + \mathcal{O}(\bm{\delta}),\quad i=1, \cdots, N.
    \end{equation*}
\end{proposition}

\subsection{Small quasi-momentum regime}\label{smallalpha}

The fixed-nonzero-quasi-momentum expansion is not uniform as $\alpha\rightarrow 0$. Consequently, the quasi-periodic Green's function develops a singular low-frequency behavior, and the small quasi-momentum regime must be treated by a separate regularized asymptotic analysis. In the regime $|\bm{\alpha}|^2<\rho\omega^2/(\lambda+2\mu)$, the quasi-momentum $|\bm\alpha|$ is bounded from above by a constant multiple of $\omega$. Hence, in this subsection, we assume that $\bm{\alpha}=\omega\bm{\alpha}_0$, where $\bm\alpha_0$ is independent of $\omega$. We focus on the case $d_l=1$ and $d_l=2$ in this subsection.

An explicit representation of $\bm{G}^{\omega\alpha_0,\omega}$ was derived in \cite{he2025uniquenessphaselessinverseelastic}. Although the formula obtained there is not directly applicable when $\omega=0$, an explicit representation of $\bm{G}^{0,0}$ can be derived by an analogous calculation. This representation is provided in equations (SM2.11) and (SM2.14) of the supplementary material. Consequently, we obtain the following asymptotic expansions. For $d_l=1$,
\begin{equation}\label{eq:greensfunctiondifference}
    \begin{aligned}
    \bm{G}^{\omega\alpha_0,\omega}(\bm{x})=\bm{G}^{0,0}(\bm{x}) +  \bm{K}_0 \ln\omega+\bm{K}_1^{\alpha_0}+\omega\ln\omega\bm P_1^{\alpha_0}(\bm x)+\omega \bm G_1^{\alpha_0}(\bm x) \\
    +\omega^2\ln\omega\bm P_2^{\alpha_0}(\bm x)+\omega^2\bm G_2^{\alpha_0}(\bm x)+\mathcal{O}(\omega^3(1+|\log\omega|)).
    \end{aligned}
\end{equation}
For $d_l=2$, 
\begin{equation}\label{eq:greensfunctiondifference2}
 \bm G^{\omega\bm\alpha_0,\omega}(\bm x)=\frac{\bm B^{\bm\alpha_0}}{\omega}
 +\bm G^{0,0}(\bm x)+\bm H^{\bm\alpha_0}(\bm x)+\omega \bm G_1^{\bm\alpha_0}(\bm x)+\omega^2\bm G_2^{\bm\alpha_0}(\bm x)+O(\omega^3).
\end{equation}
The explicit definitions of $\bm K_0, \bm K_1^{\alpha_0}, \bm P_1^{\alpha_0}$ are in (SM2.12) and (SM2.13) of supplementary material, and $\bm B^{\bm\alpha_0}, \bm H^{\bm\alpha_0}$ are defined in (SM2.15) and (SM2.16).

\begin{lemma}
For $\bm\eta\in H^{1/2}(\partial D)^3$ and $\bm q\in\mathbb C^3$, the system
\begin{equation}\label{eq:constrained}
 \mathcal S_D^{0,0}[\bm\phi]+\bm a=\bm\eta,
 \qquad \int_{\partial D}\bm\phi\mathrm{d}\sigma=\bm q
\end{equation}
has a unique solution in $H^{-1/2}(\partial D)^3\times\mathbb C^3$, depending boundedly on the data. Here $\bm a$ is a constant boundary field.
\end{lemma}
\begin{proof}
Define $\mathcal T:(H^{-1/2}(\partial D))^3\times\mathbb{C}^3\rightarrow (H^{1/2}(\partial D))^3\times\mathbb{C}^3$ by $\mathcal T(\phi,\bm a):=(\mathcal S_D^{0,0}[\phi]+\bm a, \int_{\partial D}\phi\mathrm{d}\sigma)$.
The difference between $\bm{G}^0$ and $\bm{G}^{0,0}$ is smooth on $\partial D\times\partial D$, $\mathcal S_D^{0,0} - \mathcal S_D^0$ is compact from $H^{-1/2}(\partial D)^3$ to $H^{1/2}(\partial D)^3$. The remaining difference between $\mathcal T$ and $\mathcal T_0(\phi,\bm a):=(\mathcal S_D^0[\phi], \bm a)$ is finite rank. Hence $\mathcal T-\mathcal T_0$ is compact. Since $\mathcal T_0$ is invertible, $\mathcal T$ is Fredholm of index zero.

Suppose $(\bm\phi_0,\bm a_0)$ solves \eqref{eq:constrained} with $\bm\eta=0, \bm q=0$. Set $\bm v=\mathcal S_D^{0,0}[\bm\phi]$. When $|\bm x|\rightarrow\infty$, $\bm v(\bm x)=O(r^{-1})$ and $\nabla\bm v=O(|\bm x_0|^{-2})$ for $d_l=1$. For $d_l=2$, since $\bm q=0$, 
\begin{align*}
\bm v(\bm x)= \mp\frac{1}{2}\operatorname{diag}(\mu^{-1},\mu^{-1},(\lambda+2\mu)^{-1})\int_{\partial D}y_3\bm\phi(\bm y)\mathrm{d}\sigma(\bm y)+\mathcal{O}(e^{-\kappa|x_3|}).
\end{align*}
Green's identity then gives, for both $d_l=1,2$ cases,
$$
 \int_Y\left(\lambda|\nabla\cdot\bm v|^2+
       \frac\mu2|\nabla\bm v+\nabla\bm v^\top|^2\right)d\bm x
 =-\int_{\partial D}\bm\phi\cdot\overline{\bm v}\,d\sigma=0.
$$
Hence $\bm v$ is a bounded global rigid field, thus a constant. The limit of $\bm v$ when $x_3\rightarrow\infty$ forces $\bm v=0$. The jump of conormal derivative across $\partial D$ then yields $\bm\phi=0$, and then $\bm a=0$. This finishes the proof of $\operatorname{dim}\operatorname{Ker}\mathcal T=0$. The Fredholm alternative completes the proof.
\end{proof}

Define $\mathcal S_0:H^{1/2}(\partial D)^3\rightarrow H^{-1/2}(\partial D)^3$ as $\mathcal S_0[\bm\eta]:=\bm\phi$ in \eqref{eq:constrained} for a given $\bm\eta$ and $\bm q=0$. In particular,
\begin{equation*}
 \int_{\partial D}\mathcal S_0[\bm\eta]\,d\sigma=0,
 \qquad \mathcal S_0[\bm a]=0\quad(\bm a\hbox{ constant}).
\end{equation*}

For $d_l=1$, from \eqref{eq:greensfunctiondifference}, $\mathcal S_D^{\omega\bm\alpha_0,\omega}$ has the following expansion:
\begin{align*}
 \mathcal S_D^{\omega\bm\alpha_0,\omega}[\bm\phi]
 &=\mathcal S_D^{0,0}[\bm\phi]+(\bm{K}_0 \ln\omega+\bm{K}_1^{\alpha_0})\int_{\partial D}\bm\phi\mathrm{d}\sigma+\omega\ln\omega\mathcal{R}_\omega[\bm\phi].
\end{align*}

For $d_l=2$, from \eqref{eq:greensfunctiondifference2}, $\mathcal S_D^{\omega\bm\alpha_0,\omega}$ has the following expansion:
\begin{align*}
 \mathcal S_D^{\omega\bm\alpha_0,\omega}[\bm\phi]
 &=\frac{\bm B^{\bm\alpha_0}}{\omega}\int_{\partial D}\bm\phi\,d\sigma
       +\mathcal S_D^{0,0}[\bm\phi]
       +\bm H^{\bm\alpha_0}(\bm x)\int_{\partial D}\bm\phi\,d\sigma\notag\\
 &\quad-\int_{\partial D}\bm H^{\bm\alpha_0}(\bm y)\bm\phi(\bm y)\,d\sigma
       +\omega\mathcal R_1^{\bm\alpha_0}[\bm\phi]+O(\omega^2)\bm\phi.
\end{align*}
The remainder operators $\mathcal{R}_\omega, \mathcal R_1^{\bm\alpha_0}:H^{-1/2}(\partial D)^3\rightarrow H^{1/2}(\partial D)^3$ are bounded. 

\begin{lemma}\label{lem:holomorphic}
For $d_l=1,2$, the full inverse $(\mathcal S_D^{\omega\bm\alpha_0,\omega})^{-1}$ exists for sufficiently small nonzero $\omega$, is uniformly bounded, and converges in operator norm to $\mathcal S_0$ as $\omega\rightarrow 0$. When $d_l=2$, the inverse extends holomorphically to $\omega=0$.
\end{lemma}
\begin{proof}
For $d_l=1$, solving $\mathcal S_D^{\omega\alpha_0,\omega}[\bm\phi] = \bm\eta$ with given $\bm\eta$ is equivalent to
\begin{equation}\label{eq:inverse-sd}
 \mathcal S_D^{0,0}[\bm\phi]+\bm a+\mathcal R_\omega[\bm\phi]=\bm\eta,
 \qquad \int_{\partial D}\bm\phi\mathrm{d}\sigma=(K_0\ln\omega+K_1^{\alpha_0})^{-1}\bm a,
\end{equation}
where $\bm a:=(K_0\ln\omega+K_1^{\alpha_0})\int_{\partial D}\bm\phi\mathrm{d}\sigma$. Since $\|\mathcal R_\omega\|=\mathcal O(|\omega\ln\omega|)$, this is a small perturbation of \eqref{eq:constrained} with $\bm q=0$.

For $d_l=2$, write $\bm q=\int_{\partial D}\bm\phi$ and
$\bm a=\omega^{-1}\bm B^{\bm\alpha_0}\bm q-\int_{\partial D}\bm H^{\bm\alpha_0}(\bm y)\bm\phi(\bm y)\mathrm{d}\sigma$.
Then solving $\mathcal S_D^{\omega\bm\alpha_0,\omega}[\bm\phi] = \bm\eta$ with given $\bm\eta$ becomes
\begin{equation}\label{eq:sd-inverse-dl2}
 \mathcal S_D^{0,0}[\bm\phi]+\bm a+\bm H^{\bm\alpha_0}(\bm x)\bm q
       +\omega\mathcal R_1^{\bm\alpha_0}[\bm\phi]+\mathcal O(\omega^2)\bm\phi=\bm\eta.
\end{equation}
Notice that $\bm q=\omega(\bm B^{\bm\alpha_0})^{-1}\left(\bm a+\int_{\partial D}\bm H^{\bm\alpha_0}(\bm y)\bm\phi(\bm y)\mathrm{d}\sigma\right)=\mathcal O(\omega)$. This implies that \eqref{eq:sd-inverse-dl2} is a holomorphic perturbation of \eqref{eq:constrained}. The Neumann series proves the assertions in the lemma. The inverse fails to extend holomorphically to $\omega=0$ in $d_l=1$ case because of the inverse-logarithmic terms.
\end{proof}

Similar to the large quasi-momentum case, we want to solve $\mathcal S_D^{\omega\bm\alpha_0,\omega}[\bm\phi]=\bm\psi_i^{(j)}$, which is equivalent to solving \eqref{eq:inverse-sd} with $\bm\eta=\bm\psi_i^{(j)}$. For $d_l=1$ case, we seek asymptotic expansions of $(\bm\phi, \bm a)$ in the form of 
\begin{equation}\label{eq:expansion-dl1}
\begin{aligned}
    \bm\phi &= \bm f_{i,0}^{(j)} + (\ln\omega)^{-1}\bm f_{i,1}^{(j)}+ \mathcal O(|\ln\omega|^{-2}+|\omega\ln\omega|), \\
    \bm a&=\bm a_{i,0}^{(j)} + (\ln\omega)^{-1}\bm a_{i,1}^{(j)} + \mathcal O(|\ln\omega|^{-2}+|\omega\ln\omega|),
\end{aligned}
\end{equation}

When $d_l=2$, 
\begin{equation}\label{eq:expansion-dl2}
\begin{aligned}
    \bm\phi &= \bm f_{i,0}^{(j)} + \omega\bm f_{i,1}^{(j),\bm\alpha_0}+ \mathcal{O}(\omega^2), \\
    \bm a&=\bm a_{i,0}^{(j)} + \omega\bm a_{i,1}^{(j),\bm\alpha_0} + \mathcal{O}(\omega^2),
\end{aligned}
\end{equation}

For $d_l=1,2$, the leading order term $(\bm f_{i,0}^{(j)}, \bm a_{i,0}^{(j)})$ is the solution of
\begin{equation*}
 \mathcal S_D^{0,0}[\bm f_{i,0}^{(j)}]+\bm a_{i,0}^{(j)}=\bm\psi_i^{(j)},
 \qquad\int_{\partial D}\bm f_{i,0}^{(j)}\mathrm{d}\sigma=0.
\end{equation*}
Then $\bm f_{i,0}^{(j)}=\mathcal S_0[\bm\psi_i^{(j)}]$. Now we investigate the first-order terms in \eqref{eq:expansion-dl1}. For $d_l=1$, $(\bm f_{i,1}^{(j)}, \bm a_{i,1}^{(j)})$ satisfies
\begin{equation*}
 \mathcal S_D^{0,0}[\bm f_{i,1}^{(j)}]+\bm a_{i,1}^{(j)}=0,
 \qquad\int_{\partial D}\bm f_{i,1}^{(j)}\mathrm{d}\sigma=\bm K_0^{-1}\bm a_{i,0}^{(j)}.
\end{equation*}

For $d_l=2$, put
$$
 \bm q_{i,1}^{(j)}:=(\bm B^{\bm\alpha_0})^{-1}\left(\bm a_{i,0}^{(j)}+\int_{\partial D}\bm H^{\bm\alpha_0}(\bm y)\bm f_{i,0}^{(j)}(\bm y)\mathrm{d}\sigma\right).
$$
The first correction is determined by
\begin{equation*}
\begin{aligned}
 \mathcal S_D^{0,0}[\bm f_{i,1}^{(j),\bm\alpha_0}]
          +\bm a_{i,1}^{(j),\bm\alpha_0}
   &=-\bm H^{\bm\alpha_0}(\bm x)\bm q_{i,1}^{(j)}
                       -\mathcal R_1^{\bm\alpha_0}[\bm f_{i,0}^{(j)}],\\
 \int_{\partial D}\bm f_{i,1}^{(j),\bm\alpha_0}\mathrm{d}\sigma&=\bm q_{i,1}^{(j)}.
\end{aligned}
\end{equation*}

At nonzero frequency, we retain the capacitance definition with the full inverse
\begin{equation*}
 \widehat{\mathfrak C}_{6(i-1)+j,6(i'-1)+j'}:=\widehat{\bm C}_{i,i'}^{(j,j')}(\omega)=-\int_{\partial D}(\mathcal S_D^{\omega\bm\alpha_0,\omega})^{-1}[\bm\psi_i^{(j)}]
                         \cdot\bm\psi_{i'}^{(j')}\mathrm{d}\sigma.
\end{equation*}

From \eqref{eq:expansion-dl1} and \eqref{eq:expansion-dl2}, we know that when $d_l=1$, 
$$
    (\mathcal S_D^{\omega\alpha_0,\omega})^{-1}[\bm\psi_i^{(j)}]=\bm f_{i,0}^{(j)}+\bm f_{i,1}^{(j)}(\log\omega)^{-1} +O(|\log\omega|^{-2}+|\omega\log\omega|),
$$
and when $d_l=2$,
$$
    (\mathcal S_D^{\omega\alpha_0,\omega})^{-1}[\bm\psi_i^{(j)}]=\bm f_{i,0}^{(j)}+\bm f_{i,1}^{(j),\alpha_0}\omega +O(\omega^2).
$$
We assemble the leading-order and first-order term of $\widehat{\mathfrak C}$ by
\begin{equation}\label{eq:c0-c1}
\begin{aligned}
 (\widehat{\mathfrak C}_0)_{6(i-1)+j,6(i'-1)+j'}& := -\int_{\partial D}\bm f_{i,0}^{(j)}\cdot\bm\psi_{i'}^{(j')}\mathrm{d}\sigma, \\
 (\widehat{\mathfrak C}_1)_{6(i-1)+j,6(i'-1)+j'}
   &:=\bm a_{i,0}^{(j)}\cdot\bm K_0^{-1}\bm a_{i',0}^{(j')},\\
 (\widehat{\mathfrak C}_1^{\bm\alpha_0})_{6(i-1)+j,6(i'-1)+j'}
   &:=-\int_{\partial D}\bm f_{i,1}^{(j),\bm\alpha_0}
                          \cdot\bm\psi_{i'}^{(j')}\mathrm{d}\sigma.
\end{aligned}
\end{equation}
It follows that
\begin{equation}\label{eq:capacitance-corrections}
\begin{aligned}
 d_l=1:\quad \widehat{\mathfrak C}(\omega)
  &=\widehat{\mathfrak C}_0+
        \frac{\widehat{\mathfrak C}_1}{\log(1/\omega)}
        +O\bigl(|\ln\omega|^{-2}+|\omega\ln\omega|\bigr),\\
 d_l=2:\quad \widehat{\mathfrak C}(\omega)
  &=\widehat{\mathfrak C}_0+\omega\widehat{\mathfrak C}_1^{\bm\alpha_0}
        +O(\omega^2).
\end{aligned}
\end{equation}

Define
\begin{equation*}
 \bm v_{i,0}^{(j)}(\bm x):=\mathcal S_D^{0,0}[\bm f_{i,0}^{(j)}](\bm x)+\bm a_{i,0}^{(j)}
\end{equation*}
as the static elastic field which coincides with rigid motion $\psi_i^{(j)}$ on $\partial D$.

\begin{theorem}
For $d_l=1,2$, $\widehat{\mathfrak C}_0$ is real symmetric and positive semi-definite. Its kernel consists exactly of the three common translations of all inclusions. Thus zero has geometric and algebraic multiplicity three and $\operatorname{rank}\widehat{\mathfrak C}_0=6N-3$.
\end{theorem}
\begin{proof}
For $\bm v=\sum_{i,j}c_i^{(j)}\bm v_{i,0}^{(j)}$, the interior fields are rigid motions and $\displaystyle\frac{\partial\bm v}{\partial\bm\nu}\bigg|_+ = \sum_{i,j}c_i^{(j)}\bm f_{i,0}^{(j)}$. Let $\bm{c}:=(c_i^{(j)})$. Green's identity gives
\begin{equation}\label{eq:cap0energy}
 \bm c^*\widehat{\mathfrak C}_0^\top\bm c
 =\int_{Y\setminus\overline D}
   \left(\lambda|\nabla\cdot\bm v|^2+
       \frac\mu2|\nabla\bm v+\nabla\bm v^\top|^2\right)d\bm x\ge0.
\end{equation}
\eqref{eq:cap0energy} also shows the symmetry of $\widehat{\mathfrak C}_0$. Furthermore, $\bm c^*\widehat{\mathfrak C}_0^\top\bm c=0$ implies that $\bm v$ is rigid motion also in $Y\backslash D$, hence in $Y$. The radiation condition at infinity excludes rotation, so $\bm v(\bm x)$ must be constant vector. There are exactly three linearly independent constant vectors in $\mathbb{R}^3$, which consists of the eigenspace of eigenvalue zero of $\widehat{\mathfrak C}_0$. 
\end{proof}

\begin{theorem}
Let $\alpha = \omega\alpha_0$, where $\alpha_0$ is independent of $\omega$ and satisfies $|\alpha_0|^2<\rho/(\lambda+2\mu)$.
Let $\mathfrak{M}\in\mathbb{R}^{6N\times 6N}$ be defined as in Theorem~\ref{thm:resonantfreq}, $\widehat{\mathfrak{C}}_0\in\mathbb{C}^{6N\times 6N}$ be defined as in \eqref{eq:c0-c1}. Then $\mathfrak{M}$ is positive definite,  $\widehat{\mathfrak{C}}_0$ is positive semi-definite. The kernel dimension of $\widehat{\mathfrak{C}}_0$ is 3.

Define $\uptau:=\mathrm{diag}(\tau_1^2,\cdots, \tau_N^2)\in\mathbb{R}^{6N\times 6N}$ and $\updelta:=\mathrm{diag}(\delta_1, \cdots, \delta_N)\in\mathbb{R}^{6N\times 6N}$, where each $\tau_i^2$ and $\delta_i$ is repeated six times along the diagonal. Let $\lambda^{\alpha}_{k,0}$, $k=1,\ldots,6N$, denote the eigenvalues of the matrix  $\widehat{\mathfrak{C}}_0\updelta(\mathfrak{M}\uptau)^{-1}$. The eigenvalues satisfy
$$
    0=\lambda^\alpha_{1,0}=\lambda^\alpha_{2,0}=\lambda^\alpha_{3,0}<\lambda^\alpha_{4,0}\leq\cdots\leq\lambda^\alpha_{6N,0}.
$$
The subwavelength resonant frequencies $\omega_{k,0}^{\alpha} = \omega_{k,0}^{\alpha}(\bm{\delta})$ associated with $\mathcal{A}_{\bm{\delta}}^{\omega\alpha_0,\omega}$ satisfy
\begin{equation*}
\omega^{\alpha}_{k,0} := (\lambda^{\alpha}_{k,0}/\rho)^{1/2} +
\left\{
\begin{aligned}
    &\mathcal{O}(\sqrt{\delta_M}/|\ln\delta_M|),&& d_l=1,\\
    &\mathcal{O}(\delta_M), && d_l=2,
\end{aligned}
\right.
\quad k=4, \ldots, 6N.
\end{equation*}

Let $m_\delta:=\sum_i\tau_i^2|D_i|/\delta_i$. For $d_l=1$, let $\nu_j(\bm\delta)$ be the eigenvalues of $\bm K_0^{-1}/(\rho m_\delta)$. Three further nonzero resonances satisfy
\begin{equation}\label{eq:resonantfrequency-dl1}
 \omega_j=\sqrt{\frac{2\nu_j(\bm\delta)}{|\ln\delta_M|}}
 \left[1+\mathcal O\left(\frac{\ln|\ln\delta_M|}{|\ln\delta_M|}\right)\right],
 \qquad j=1,2,3.
\end{equation}
For $d_l=2$, let $\zeta_j(\bm\delta)$ be the eigenvalues of $-(\bm B^{\alpha_0})^{-1}/(\rho m_\delta)$. Three further nonzero resonances satisfy
\begin{equation*}
 \omega_j=\zeta_j(\bm\delta)+\mathcal O(\delta_M^2),
 \qquad j=1,2,3.
\end{equation*}
Notice that those three eigenvalues are complex, which corresponds to the evanescent eigenmode.
\end{theorem}
\begin{proof}
    When $\omega$ is a resonant frequency, the following system admits a nontrivial normalized solution $\bm{\psi}, \bm{\phi}\in (H^{-1/2}(\partial D))^3$:
    \begin{equation}\label{eq:bandsys1}
    \begin{cases}
    \mathcal S_{D_i}^{\tau_i\omega}[\bm\psi_i]
       -\mathcal S_D^{\omega\bm\alpha_0,\omega}[\bm\phi]=0,\\[1mm]
    \left(-\frac12\mathcal I+(\mathcal K_{D_i}^{\tau_i\omega})^*\right)[\bm\psi_i]
       -\delta_i\left(\tfrac12\mathcal I+
          (\mathcal K_D^{\omega\bm\alpha_0,\omega})^*\right)[\bm\phi]=0
                    \quad\text{on }\partial D_i.
    \end{cases}
    \end{equation}
    
    The factor $\delta_i$ multiplies the exterior traction on the observation inclusion. Substitution of the correct operator expansions gives, for $d_l=1$,
    \begin{equation*}
    \begin{aligned}
    &\left[-\dfrac12\mathcal I+(\mathcal K_{D_i}^0)^*+\tau_i^2\omega^2(\mathcal K_{D_i,2}^0)^*+\tau_i^3\omega^3(\mathcal K_{D_i,3}^0)^*+\mathcal O(\omega^4)\right][\bm\psi_i]\\
    &\quad=\delta_i\bigl[\dfrac12\mathcal I+(\mathcal K_D^{0,0})^*+\omega\log\omega(\mathcal K_{D,1,\log}^{\alpha_0})^*+\omega(\mathcal K_{D,1}^{\alpha_0})^*\\
    &\qquad\quad+\omega^2\log\omega(\mathcal K_{D,2,\log}^{\alpha_0})^*+\omega^2(\mathcal K_{D,2}^{\alpha_0})^*+\mathcal O\bigl(\omega^3(1+|\log\omega|)\bigr)\bigr][\bm\phi].
    \end{aligned}
    \end{equation*}
    For $d_l=2$, the right-hand bracket is instead
    \begin{equation*}
    \dfrac12\mathcal I+(\mathcal K_D^{0,0})^*+(\mathcal K_{D,H}^{\bm\alpha_0})^*+\omega(\mathcal K_{D,1}^{\bm\alpha_0})^*+\omega^2(\mathcal K_{D,2}^{\bm\alpha_0})^*+\mathcal O(\omega^3).
    \end{equation*}

    Letting $\omega\to 0$ and $\bm{\delta}\to 0$ in the second equation of \eqref{eq:bandsys1}, we conclude that the leading-order term of $\bm{\psi}_i$ lies in the span of ${\bm{f}_i^{(j)}}$, which forms a basis of $\operatorname{Ker}(-\frac{1}{2}\mathcal{I}+(\mathcal{K}_D^0)^*)$. Consequently, we may write $\bm{\psi}_i = \sum_{i,j} \tilde{c}_i^{(j)}\bm{f}_i^{(j)} + \mathcal{O}(\omega^2+\bm{\delta})$ on $\partial D_i$ with constants $\tilde{c}_i^{(j)}$.

    Taking the inner product on $\partial D$ of both sides of the second equation in \eqref{eq:bandsys1} with $\bm{\psi}_{i'}^{j'}$. On one hand, we have
    \begin{equation*}
    \begin{aligned}
        &\int_{\partial D_i}\left(-\frac{1}{2}\mathcal I+(\mathcal K_{D_i}^{\tau_i\omega})^*\right)[\bm\psi_i]\cdot\bm\psi_i^{(j')}\mathrm{d}\sigma \\
        =&\tau_i^2\omega^2\int_{\partial D_i}(\mathcal K_{D_i,2}^0)^*[\bm\psi_i]\cdot\bm\psi_i^{(j')}\mathrm{d}\sigma + \tau_i^3\omega^3\int_{\partial D_i}(\mathcal K_{D_i,3}^0)^*[\bm\psi_i]\cdot\bm\psi_i^{(j')}\mathrm{d}\sigma + \mathcal{O}(\omega^4) \\
        =&-\rho\tau_i^2\omega^2\int_{D_i}(\mathcal S_{D_i}^0+\tau_i\omega\mathcal S_{D_i,1}^0)[\bm\psi_i]\cdot\bm\psi_i^{(j')}\mathrm{d}\bm x+\mathcal O(\omega^4),
    \end{aligned}
    \end{equation*}
    where the last equality uses Lemma SM2.2 in the supplementary material. Combining with the expression of $\bm\psi_{i'}^{j'}$ and using \eqref{eq:potentialsingle} we have
    \begin{equation*}
    \begin{aligned}
        (\mathcal S_{D_i}^0+\tau_i\omega\mathcal S_{D_i,1}^0)[\bm\psi_i] =&\sum\limits_{j=1}^6 \widetilde{c}_i^{(j)}\bm\psi_i^{(j)} + \tau_i\omega\sum\limits_{j=1}^6\widetilde{c}_i^{(j)}\mathcal S_{D_i,1}^0(\mathcal S_{D_i}^0)^{-1}[\bm\psi_i^{(j)}] + \mathcal O(\omega^2+\bm\delta)\|\widetilde{\bm{c}}\|, \\
        =:&\sum\limits_{j=1}^6 c_i^{(j)}\bm\psi_i^{(j)} + \mathcal O(\omega^2+\bm\delta)\|\widetilde{\bm c}\|,
    \end{aligned}    
    \end{equation*}
    where for $\ell=1,2,3$,
    $$
        c_i^{(\ell)}:=\widetilde c_i^{(\ell)}-\frac{\ii\tau_i\omega}{12\pi\rho}(\frac{2}{c_s^3}+\frac{1}{c_p^3})\sum_{j=1}^6\widetilde c_i^{(j)}\left[\int_{\partial D_i}(\mathcal S_{D_i}^0)^{-1}[\bm\psi_i^{(j)}]d\sigma\right]_{\ell},
    $$
    and $c_i^{(l)} := \widetilde c_i^{(l)}$ for $l=4,5,6$. Hence
    \begin{equation}\label{eq:omega0-2}
    \begin{aligned}
        \frac1{\delta_i}\int_{\partial D_i}\left(-\dfrac12\mathcal I+(\mathcal K_{D_i}^{\tau_i\omega})^*\right)[\bm\psi_i]\cdot\bm\psi_i^{(j')}\mathrm d\sigma
        =-\frac{\rho\tau_i^2\omega^2}{\delta_i}\sum_j\bm M_{i,i}^{(j,j')}c_i^{(j)}+\mathcal{O}\left(|\omega|^2+\frac{|\omega|^4}{\bm\delta}\right)\|\bm c\|.
    \end{aligned}
    \end{equation}

    On the other hand, from \eqref{eq:expansionsingle} and \eqref{eq:potentialsingle}, we obtain the expansion
    \begin{equation*}
        \mathcal{S}_D^{\bm{\tau}\omega}[\bm{f}_i^{(j)}] = \bm{\psi}_i^{(j)} - \sum\limits_{k=1}^N \tau_k\omega\frac{\ii}{12\pi\rho}\Big(\frac{2}{c_s^3}+\frac{1}{c_p^3}\Big)\int_{\partial D_k}\bm{f}_i^{(j)}(\bm{y})\mathrm{d}\sigma(\bm{y})\chi_{\partial D_k}+\mathcal{O}(\omega^2),
    \end{equation*}
    and hence
    \begin{equation*}
    \begin{aligned}
        &\mathcal{S}_D^{\omega\alpha_0,\omega}[\bm{\phi}] = \mathcal{S}_D^{\bm{\tau}\omega}[\bm{\psi}] = \sum\limits_{i,j} \widetilde{c}_i^{(j)}\mathcal{S}_D^{\bm{\tau}\omega}[\bm{f}_i^{(j)}] + \mathcal{O}(\omega^2+\bm{\delta}) \\
        &= \sum\limits_{i,j} \widetilde{c}_i^{(j)}\bigg(\bm{\psi}_i^{(j)} - \frac{\ii}{12\pi\rho}\Big(\frac{2}{c_s^3}+\frac{1}{c_p^3}\Big)\sum\limits_{k=1}^N \tau_k\omega\int_{\partial D_k} \bm{f}_i^{(j)}(\bm{y})\mathrm{d}\sigma(\bm{y})\chi_{\partial D_k}\bigg) +\mathcal{O}(\omega^2+\bm{\delta}).
    \end{aligned}
    \end{equation*}
    This implies that
    \begin{equation*}
        \bm\phi=\sum_{i,j}c_i^{(j)}(\mathcal S_D^{\omega\bm\alpha_0,\omega})^{-1}[\bm\psi_i^{(j)}]
        +\mathcal{O}(|\omega|^2+\bm\delta)\|\bm c\|.
    \end{equation*}
    
    Notice that $\bm{\psi}_{i}^{(j)}\in\mathrm{Ker}(-\frac{1}{2}\mathcal{I}+\mathcal{K}_D^{0,0})$. Therefore,
    \begin{equation}\label{eq:omega0-3}
        \int_{\partial D} (\frac{1}{2}\mathcal I+(\mathcal{K}_D^{\omega\bm\alpha_0,\omega})^*)[\bm{\phi}]\cdot\bm{\psi}_{i'}^{(j')}\mathrm{d}\sigma =-\sum\limits_{i,j}\widehat{\bm C}_{i,i'}^{(j,j')}(\omega)c_i^{(j)}+\mathcal{O}(|\omega|^2+\bm\delta)\|\bm c\|.
    \end{equation}
 
    Combining equations \eqref{eq:omega0-2} and \eqref{eq:omega0-3} we conclude that the resonance frequency $\omega$ must satisfy
    $$
       \operatorname{det}\left[(\widehat{\mathfrak{C}}(\omega) + \mathfrak{R}(\omega,\bm\delta)^\top)\updelta(\mathfrak M\uptau)^{-1}-\rho\omega^2\bm I_{6N}\right]=0,
    $$
    where $\|\mathfrak{R}(\omega,\bm{\delta})\|=\mathcal{O}(\bm{\delta}+|\omega|^2+|\omega|^4/\bm{\delta})$. Then the leading-order part of resonant frequencies are in the form of $(\lambda_{k,0}^\alpha/\rho)^{1/2}$, where $\lambda_{k,0}^\alpha$ are eigenvalues of $\widehat{\mathfrak{C}}_0\updelta(\mathfrak M\uptau)^{-1}$.
    
    For $d_l=1$, from \eqref{eq:capacitance-corrections} we know that $\widehat{\mathfrak{C}}(\omega)-\widehat{\mathfrak{C}}_0 = \mathcal{O}(|\ln\omega|^{-1})$, and $\omega = \mathcal{O}(\sqrt{\bm\delta})$, so $\omega_k = \sqrt{\lambda_{k,0}^\alpha/\rho} + \mathcal{O}(\sqrt{\bm\delta}/|\ln{\bm\delta}|)$. The result for $d_l=2$ can be concluded similarly from \eqref{eq:capacitance-corrections}.
    
    To determine the order of resonant frequencies corresponding to zero-eigenvalue of $\widehat{\mathfrak{C}}_0$, we define $T:\mathbb{C}^3\rightarrow\mathbb{C}^{6N}$ by $T\bm b:=(\bm b, \bm b, \cdots, \bm b)$ for any constant vector $\bm b\in\mathbb{C}^3$. Then we have $\operatorname{Ker}\widehat{\mathfrak{C}}_0=\operatorname{Ran}T$. Furthermore, from the definition of $\mathfrak{M}$ it follows that for any $\bm b\in\mathbb{C}^3$, 
    $$
        (T\bm b)^\top\mathfrak{M}\uptau\updelta^{-1}(T\bm b) = \left(\sum_i \frac{\tau_i^2|D_i|}{\delta_i}\right)|\bm b|^2.
    $$
    Define $m_\delta:=\sum\limits_{i=1}^N\tau_i^2|D_i|/\delta_i$, then $T^\top\mathfrak{M}\uptau\updelta^{-1}T = m_\delta I_3$. This implies that $m_\delta=\mathcal{O}(\bm\delta^{-1})$.
    
    For $d_l=1$, recall
    $$
        \widehat{\mathfrak{C}}(\omega) = \widehat{\mathfrak{C}}_0-\widehat{\mathfrak{C}}_1(\ln\omega)^{-1} + o(|\ln\omega|^{-1}).
    $$
    Furthermore, $T^\top\widehat{\mathfrak{C}}_1T=K_0^{-1}$. We then have
    $$
    T^\top\widehat{\mathfrak{C}}(\omega)T-T^\top\rho\omega^2(\mathfrak{M}\uptau)\updelta^{-1}T\bm = -K_0^{-1}(\ln\omega)^{-1} - \rho\omega^2m_\delta \bm I_3+o(|\ln\omega|^{-1}).
    $$ 
    The equation above implies that $\omega^2/\bm\delta=\mathcal{O}(|(\ln\omega)^{-1}|)$. Taking logarithm on both sides yields $\ln(1/\omega) = |\ln\bm\delta|/2+\mathcal{O}(\ln|\ln\bm\delta|)$. This proves the formula for resonant frequencies \eqref{eq:resonantfrequency-dl1}. The result for $d_l=2$ follows similarly. From (SM2.15), eigenvalues of $\bm B^{\bm\alpha_0}$ are negative-imaginary numbers, which means that $\zeta_j(\bm\delta)$ are also negative-imaginary numbers.
\end{proof}

It is worth mentioning that there is no inconsistency between the definitions of $\widehat{\bm{C}}_{i,i',0}^{(j,j')}$ and $\widehat{\bm{C}}_{i,i'}^{(j,j')}$. More precisely, from Lemma \ref{lem:holomorphic} we have the following proposition, holding for $d_l=1,2$.

\begin{proposition}
The coefficients $\widehat{\bm{C}}_{i,i'}^{(j,j')}$ for $\bm\alpha\in Y^*\backslash\{0\}$ and their periodic counterpart $\widehat{\bm{C}}_{i,i',0}^{(j,j')}$ satisfy
$$
\widehat{\bm{C}}_{i,i',0}^{(j,j')} = \lim_{\alpha \to 0}\widehat{\bm{C}}_{i,i'}^{(j,j')}.
$$
\end{proposition}

\section{Perturbed systems}

In this section, we consider the perturbed system \eqref{eq:elas_sys}. In contrast to the unperturbed periodic setting, the presence of perturbations destroys periodicity, and therefore the corresponding solution $\bm{u}$ is no longer expected to be quasi-periodic. In the first subsection, we reformulate the problem using the Floquet transform.

The main objective of this section is to investigate Anderson localization under random perturbations of the scatterers. However, analyzing systems with infinitely many random perturbations is highly challenging, both theoretically and computationally. For this reason, we first consider the case of a single scatterer per periodic cell and introduce random perturbations on $M$ scatterers. We observe that the decay rate of the eigenmodes increases as $M$ increases, suggesting stronger localization and, in the limiting case $M\to\infty$, perfect localization in the sense described by Anderson \cite{PhysRev.109.1492}.

\subsection{Floquet transform}

Define the Floquet transform of $\bm{u}$ by
$$
\mathcal{F}[\bm{u}](\bm{\alpha},\bm{x}):= \sum\limits_{\bm{m}\in\Lambda} \bm{u}(\bm{x}+\bm{m})e^{-\ii\bm{\alpha}\cdot \bm{m}}.
$$
Then $\mathcal{F}\bm{u}(\bm{\alpha},\bm{x})$ is $\bm{\alpha}$-quasi-periodic in $\bm{x}$. The inverse Floquet transform is defined by
$$
\mathcal{F}^{-1}[\bm{v}](\bm{m}, \bm{x}):=\frac{1}{|Y^*|}\int_{Y^*}\bm{v}(\bm{\alpha},\bm{x}-\bm{m})e^{\ii\bm{\alpha}\cdot\bm{m}}\mathrm{d}\alpha.
$$

Let $b_i^{\bm{m}}:=(\tau_i^{\bm{m}})^2/\tau_i^2$ be the perturbation factor at each lattice point $\bm{m}\in\Lambda$, and denote $\bm{u}^{\bm{\alpha}}(\bm{x}) := \mathcal{F}\bm{u}(\bm{\alpha},\bm{x})$. Applying the Floquet transform to \eqref{eq:elas_sys} yields the following system:
\begin{equation}\label{eq:elas_sys_floquet}
	\left\{
	\begin{aligned}
	&\displaystyle (\mathcal{L}^{\lambda, \mu} +\rho\omega^2) \bm{u}^{\bm{\alpha}}= \bm{0}  &&   \mathrm{in }\  Y \setminus\bar{D}, \medskip \\
	&\displaystyle \mathcal{L}^{\lambda, \mu} \bm{u}^{\bm{\alpha}} +\rho\tau_i^2\omega^2\sum\limits_{m\in\Lambda}b_i^{\bm{m}}\bm{u}(\bm{x}+\bm{m})e^{-\ii\bm{\alpha}\cdot \bm{m}} = \bm{0}  &&  \mathrm{in }\  D_i, i=1,\cdots, N,  \medskip\\
	&\displaystyle \bm{u}^{\bm{\alpha}}|_+ - \bm{u}^{\bm{\alpha}}|_- = \bm{0}   && \mathrm{on }\  \partial D, \medskip \\
    &\displaystyle \delta_i \partial_{\bm\nu}\bm{u}^{\bm{\alpha}}|_+ - \partial_{\bm\nu}\bm{u}^{\bm{\alpha}}|_-  = \bm{0}   &&\mathrm{on }\  \partial D_i, i=1, \cdots, N, \medskip \\
    &\bm{u}^{\bm{\alpha}}(\bm{x}+\bm{m}) = e^{\ii\bm{\alpha}\cdot \bm{m}}\bm{u}^{\bm{\alpha}}(\bm{x}) && \mathrm{for}~\mathrm{all}~\bm{m}\in\Lambda,
    \end{aligned}
	\right.
\end{equation}
together with the outgoing radiation condition for $\bm{u}(\bm{x}_l, x_0)$ in the non-periodic direction as $|x_0|\to\infty$.

Define $\bm{v}_i^{(j),\bm{\alpha}}:=\mathcal{S}_D^{\bm{\alpha},0}[\bm{f}_{i,\mathrm{per}}^{(j)}]$ on $Y\backslash D$. Then $\bm{v}_i^{(j),\bm{\alpha}}$ satisfies
\begin{equation*}
\left\{
\begin{aligned}
    \mathcal{L}^{\lambda,\mu}\bm{v}_i^{(j),\bm{\alpha}} & = 0 &&\mathrm{in}~ Y \backslash \bar{D}, \\
    \bm{v}_i^{(j),\bm{\alpha}} &= \bm{\psi}_i^{(j)} &&\mathrm{on}~\partial D, \\
    \bm{v}_i^{(j),\bm{\alpha}}(\bm{x}+\bm{l}) &= e^{\ii\bm{\alpha}\cdot \bm{l}}\bm{v}_i^{(j),\bm{\alpha}}(\bm{x}) &&\mathrm{for~ any} ~ \bm{l}\in\Lambda.
\end{aligned}
\right.
\end{equation*}
Since $\partial_{\bm\nu}\bm{u}|_- = \delta_i\partial_{\bm\nu}\bm{u}|_+,$ the solution $\bm{u}$ of \eqref{eq:elas_sys} in each inclusion $D_i^m$ is, in the limit $\delta_i \to 0$, well approximated by Neumann eigenmodes associated with $\Psi$. Hence, we expect that $\bm{u}^{\bm{\alpha}}$ can be approximated in $Y\setminus \overline{D}$ by linear combinations of $\bm{v}_i^{(j),\bm{\alpha}}$, $i=1,\ldots,N$, $j=1,\ldots,6$. We now make this statement precise.

Let $\bm{\psi}_{i,\bm{m}}^{(j)}(\bm{x}):= \bm{\psi}_i^{(j)}(\bm{x}+\bm{m})$, $\bm{m}=(m_1,m_2,m_3)\in\Lambda$. From the definition of $\bm{\psi}_i^{(j)}$, we obtain
\begin{align*}
&\bm{\psi}_{i,\bm{m}}^{(j)} = \bm{\psi}_i^{(j)}, \quad j=1,2,3, \\
&\bm{\psi}_{i,\bm{m}}^{(4)} = \bm{\psi}_i^{(4)} - m_2\bm{\psi}_i^{(1)} + m_1\bm{\psi}_i^{(2)}, \\
&\bm{\psi}_{i,\bm{m}}^{(5)} = \bm{\psi}_i^{(5)} - m_3\bm{\psi}_i^{(2)} + m_2\bm{\psi}_i^{(3)}, \\
&\bm{\psi}_{i,\bm{m}}^{(6)} = \bm{\psi}_i^{(6)} + m_3\bm{\psi}_i^{(1)} - m_1\bm{\psi}_i^{(3)}.
\end{align*}
Therefore, for each fixed $\bm{m}\in\Lambda$, $\bm{\psi}_{i,\bm{m}}^{(j)}$ can be expressed as a linear combination of ${\bm{\psi}_i^{(j)}}, j=1, \ldots, 6$. From Proposition~\ref{prop:normalalpha-field}, it follows that for each fixed $i$ and $\bm{m}$ there exist coefficients $\tilde c_{i,\bm{m}}^{(j)}$ such that, for $\bm{x}\in D_i^{\bm{m}}$,
$$
\bm{u}(\bm{x}) = \sum_{j} \tilde{c}_{i,\bm{m}}^{(j)}\bm{\psi}_{i,\bm{m}}^{(j)}(\bm{x}) + \mathcal{O}(\omega^2+\bm{\delta}) = \sum_j c_{i,\bm{m}}^{(j)}\bm{\psi}_i^{(j)}(\bm{x}-\bm{m}) + \mathcal{O}(\omega^2+\bm{\delta}).
$$
Consequently, when $\bm{x}\in D_i$, we obtain
\begin{equation*}
\begin{aligned}
\bm{u}^{\bm{\alpha}}(\bm{x}) &= \sum\limits_{\bm{m}\in\Lambda} \bm{u}(\bm{x}+\bm{m})e^{-\ii\bm{\alpha}\cdot \bm{m}} = \sum\limits_{\bm{m}\in\Lambda}\sum_j c_{i,\bm{m}}^{(j)}\bm{\psi}_{i}^{(j)}(\bm{x})e^{-\ii\bm{\alpha}\cdot \bm{m}} + \mathcal{O}(\omega^2+\bm{\delta}) \\
&= \sum_j c_{i,\bm{\alpha}}^{(j)}\bm{\psi}_i^{(j)}(\bm{x}) + \mathcal{O}(\omega^2+\bm{\delta}),
\end{aligned}
\end{equation*}
where $c_{i,\bm{\alpha}}^{(j)}:=\sum\limits_{\bm{m}\in\Lambda} c_{i,\bm{m}}^{(j)}e^{-\ii\bm{\alpha}\cdot \bm{m}}$. Hence, for $\bm{x}\in Y \backslash\overline{D}$,
\begin{align*}
    \bm{u}^{\bm{\alpha}}(x) &= \sum_{i,j} c_{i,\bm{\alpha}}^{(j)} \mathcal{S}_D^{\bm{\alpha},\omega}[\bm{f}_{i,\mathrm{per}}^{(j)}] + \mathcal{O}(\omega^2+\bm{\delta}) = \sum_{i,j} c_{i,\bm{\alpha}}^{(j)}\mathcal{S}_D^{\bm{\alpha},0}[\bm{f}_{i,\mathrm{per}}^{(j)}] + \mathcal{O}(\omega^2+\bm{\delta}) \\ &= \sum_{i,j} c_{i,\bm{\alpha}}^{(j)}\bm{v}_i^{(j),\bm{\alpha}}+\mathcal{O}(\omega^2+\bm{\delta}).
\end{align*}
As a result, we obtain the representation
\begin{equation*}
\bm{u}^{\bm{\alpha}}(\bm{x}) = \left\{
\begin{aligned}
    &\sum_{i,j} c_{i,\bm{\alpha}}^{(j)}\bm{\psi}_i^{(j)}(\bm{x}) + \mathcal{O}(\omega^2+\bm{\delta}), && \bm{x} \in D, \\
    &\sum_{i,j} c_{i,\bm{\alpha}}^{(j)}\bm{v}_i^{(j),\bm{\alpha}}(\bm{x}) + \mathcal{O}(\omega^2+\bm{\delta}), && \bm{x}\in Y\backslash\overline{D}.
\end{aligned}
\right.
\end{equation*}

Our subsequent analysis consists of two main parts. First, we identify the frequencies $\omega$ for which \eqref{eq:elas_sys_floquet} admits nontrivial solutions, i.e., the resonant frequencies. Second, we characterize the corresponding nontrivial solutions $\bm{u}$ associated with these resonant frequencies. To this end, we aim to determine all coefficients $c_{i,\bm{m}}^{(j)}$.

By Green's second identity for the Lam\'{e} system (cf. \cite[(1.50)]{ammari2015mathematical}), we obtain
\begin{align*}
&\int_{\partial D} \partial_{\bm\nu} \bm{u}^{\bm{\alpha}}|_+\cdot \bm{\psi}_{i'}^{(j')}\mathrm{d}\sigma = \sum_i\frac{1}{\delta_i}\int_{\partial D_i} \partial_{\bm\nu} \bm{u}^{\bm{\alpha}}|_-\cdot \bm{\psi}_{i'}^{(j')}\mathrm{d}\sigma \\=& \sum_i\frac{1}{\delta_i}\int_{D_i} (-\rho\tau_i^2\omega^2)\left(\sum\limits_{\bm{m}\in\Lambda} b_i^{\bm{m}} \bm{u}(\bm{x}+\bm{m}) e^{-\ii\bm{\alpha}\cdot \bm{m}}\right)\cdot\bm{\psi}_{i'}^{(j')}\mathrm{d}x \\
=& -\sum_i\frac{\omega^2\rho\tau_i^2}{\delta_i}\sum\limits_{\bm{m}\in\Lambda} b_i^{\bm{m}} e^{-\ii\bm{\alpha}\cdot \bm{m}}\sum\limits_j c_{i,\bm{m}}^{(j)}\int_{D_i}\bm{\psi}_i^{(j)}\cdot\bm{\psi}_{i'}^{(j')}\mathrm{d}x + \mathcal{O}\big(\frac{\omega^4}{\bm{\delta}}+\omega^2\big) \\
=& -\sum_{i,j}\frac{\omega^2\rho\tau_i^2}{\delta_i}\left(\sum\limits_{\bm{m}\in\Lambda} b_i^{\bm{m}} e^{-\ii\bm{\alpha}\cdot \bm{m}} c_{i,\bm{m}}^{(j)}\right) \bm{M}_{i,i'}^{(j,j')} + \mathcal{O}\big(\frac{\omega^4}{\bm{\delta}}+\omega^2\big).
\end{align*}
On the other hand,
\begin{align*}
&\int_{\partial D} \partial_{\bm\nu} \bm{u}^{\bm{\alpha}}|_+\cdot \bm{\psi}_{i'}^{(j')}\mathrm{d}\sigma = \int_{\partial D} \sum\limits_{i,j} c_{i,\bm{\alpha}}^{(j)} \partial_{\bm\nu} \bm{v}_i^{(j),\bm{\alpha}}|_+\cdot \bm{\psi}_{i'}^{(j')}\mathrm{d}\sigma +\mathcal{O}(\omega^2+\bm{\delta})\\=& \sum\limits_{i,j}c_{i,\bm{\alpha}}^{(j)}\int_{\partial D} \bm{f}_{i,\mathrm{per}}^{(j)}\cdot \bm{\psi}_{i'}^{(j')}\mathrm{d}\sigma + \mathcal{O}(\omega^2+\bm{\delta})= -\sum_{i,j} c_{i,\bm{\alpha}}^{(j)}\widehat{\bm{C}}_{i,i'}^{j,j'} + \mathcal{O}(\omega^2+\bm{\delta}).
\end{align*}
Combining the above identities yields
$$
    \sum_{i,j}\frac{\omega^2\rho\tau_i^2}{\delta_i}\left(\sum\limits_{\bm{m}\in\Lambda} b_i^{\bm{m}} e^{-\ii\bm{\alpha}\cdot \bm{m}} c_{i,\bm{m}}^{(j)}\right) \bm{M}_{i,i'}^{(j,j')} = \sum_{i,j} \sum_{\bm{m}\in\Lambda} c_{i,\bm{m}}^{(j)}e^{-\ii\bm{\alpha}\cdot \bm{m}}\widehat{\bm{C}}_{i,i'}^{j,j'}.
$$
Using the notation of Theorem~\ref{thm:resonantfreq}, this relation can be written in matrix form as
\begin{equation}\label{eq:floquetrel}
    \omega^2\mathfrak{M}^\top\uprho\sum\limits_{\bm{m}\in\Lambda}\mathfrak{B}_{
    \bm{m}}\mathfrak{c}_{\bm{m}} e^{-\ii\bm{\alpha}\cdot \bm{m}} = \widehat{\mathfrak{C}}_{\bm{\alpha}}^\top\sum\limits_{\bm{m}\in\Lambda} \mathfrak{c}_{\bm{m}} e^{-\ii\bm{\alpha}\cdot \bm{m}}.
\end{equation}
Here $\uprho:=\mathrm{diag}(\rho\tau_1^2/\delta_1,\cdots, \rho\tau_N^2/\delta_N)$ and $\mathfrak{B}_{\bm{m}}:=\mathrm{diag}(b_1^{\bm{m}},\cdots, b_N^{\bm{m}})$ are $6N\times 6N$ diagonal matrices, where each entry is repeated six times along the diagonal. The vector $\mathfrak{c}_{\bm{m}}\in\mathbb{R}^{6N}$ is defined by $(\mathfrak{c}_{\bm{m}})_{6(i-1)+j}:=c_{i,\bm{m}}^{(j)}$. We write $\widehat{\mathfrak{C}}_\alpha^\top$ instead of $\widehat{\mathfrak{C}}^\top$ here to emphasize its dependence upon $\bm{\alpha}$.

Define the discrete Floquet transform and its inverse
$$\mathcal{F}_D:l^2(\Lambda,\mathbb{C}^{6N}) \rightarrow L^2(Y^*, \mathbb{C}^{6N}),\quad \mathcal{F}^{-1}_D:L^2(Y^*,\mathbb{C}^{6N})\rightarrow l^2(\Lambda,\mathbb{C}^{6N})$$
by
$$
\mathcal{F}_D[\bm{\phi}](\bm{\alpha}):=\sum\limits_{\bm{m}\in\Lambda}\bm{\phi}(\bm{m})e^{-\ii\bm{\alpha}\cdot \bm{m}}, \quad \mathcal{F}^{-1}_D[\bm{\psi}](\bm{m}):=\frac{1}{|Y^*|}\int_{Y^*}\bm{\psi}(\bm{\alpha})e^{\ii\bm{\alpha}\cdot \bm{m}}\mathrm{d}\bm{\alpha}.
$$
Viewing $\mathfrak{c}_{\bm{m}}$ as a sequence in $\ell^2(\Lambda,\mathbb{C}^{6N})$, the left-hand side of \eqref{eq:floquetrel} can be written as $\omega^2\mathfrak{M}^\top\uprho\mathcal{F}_D[\mathfrak{B}\mathfrak{c}](\bm{\alpha})$,
while the right-hand side becomes $\widehat{\mathfrak{C}}_{\bm{\alpha}}^\top\mathcal{F}_D[\mathfrak{c}]$. We further define $\mathfrak{C}_n:=\mathcal{F}_D^{-1}[\widehat{\mathfrak{C}}_\alpha^\top](\bm{n})$. Applying the inverse Floquet transform to both sides of \eqref{eq:floquetrel} yields
\begin{equation}\label{eq:resfreqrel}
    \omega^2\mathfrak{M}^\top\uprho\mathfrak{B}_{\bm{n}}\mathfrak{c}_{\bm{n}} = \sum\limits_{\bm{m}\in\Lambda}\mathfrak{C}_{\bm{n}-\bm{m}}\mathfrak{c}_{\bm{m}},
\end{equation}
which is an infinite system of coupled equations for the coefficients ${\mathfrak{c}_{\bm{m}}}, {\bm{m}\in\Lambda}$, and characterizes the resonant modes. The values of $\omega$ for which \eqref{eq:resfreqrel} admits nontrivial solutions correspond to the resonant frequencies of \eqref{eq:elas_sys_floquet}.

Define the operator
$\mathcal{C}:l^2(\Lambda, \mathbb{C}^{6N})\rightarrow l^2(\Lambda, \mathbb{C}^{6N})$ by
$$
    \mathcal{C}[\bm{p}](\bm{m}) = \sum\limits_{\bm{m}\in\Lambda}\mathfrak{C}_{\bm{n}-\bm{m}}\bm{p}_{\bm{m}}.
$$
Then its Floquet transform
$$
\mathcal{M}_{\widehat{\mathfrak{C}}}:=\mathcal{F}_D\mathcal{C}\mathcal{F}_D^{-1}: l^2(Y^*, \mathbb{C}^{6N})\rightarrow l^2(Y^*, \mathbb{C}^{6N})
$$
is a multiplication operator, since $(\mathcal{M}_{\widehat{\mathfrak{C}}}\bm{f})(\bm{\alpha}) = \widehat{\mathfrak{C}}^{\bm{\alpha}}\bm{f}(\bm{\alpha})$. It follows that
$$
\sigma(\mathcal{C}) = \sigma(\mathcal{M}_{\widehat{\mathfrak{C}}}) = \cup \{\lambda_j(\bm{\alpha}):\bm{\alpha}\in Y^*\},
$$
where $\lambda_j(\bm{\alpha})$ denote the eigenvalues of $\widehat{\mathfrak{C}}^{\bm{\alpha}}$.

When the perturbation is absent in \eqref{eq:resfreqrel}, we have $\mathfrak{B}_{\bm{n}}=\bm{I}$. In this case, the previous discussion implies that the spectrum of $\mathcal{C}$ is purely absolutely continuous, and no point spectrum occurs (except in the presence of flat bands). After introducing perturbations, $\mathfrak{B}_{\bm{n}}\neq \bm{I}$, and \eqref{eq:resfreqrel} may admit eigenvalues, which correspond to localized modes of interest. Such eigenvalues cannot lie in the absolutely continuous spectrum; hence they must be located in spectral gaps. Theorem 4.4 of \cite{doi:10.1137/S0036139997320536} further shows that eigenmodes induced by compact perturbations and lying in spectral gaps exhibit exponential decay.

We emphasize that localization induced by compact perturbations, while leading to exponentially decaying modes, does not correspond to Anderson localization. In contrast, Anderson localization refers to genuinely localized eigenmodes arising in the limiting regime of an infinite system with random perturbations applied to all scatterers. Since such infinite random perturbations are not directly accessible, we instead employ a key qualitative feature of Anderson localization: in the unperturbed periodic system all modes are propagating, whereas increasing disorder leads to progressively stronger localization of eigenmodes. In the next section, we first establish localization induced by compact perturbations, and then investigate the emergence of Anderson localization by gradually increasing the strength and spatial extent of the perturbations.

\subsection{Compact perturbation}

We now discuss the system with compact perturbations in detail. We focus on chain-like periodic structures, i.e., the case $d_l=1$; the cases $d_l=2$ and $d_l=3$ can be treated in a similar manner.

Let $\Lambda=\mathbb{Z}$ be the lattice. The corresponding Brillouin zone is $Y^* = (-\pi, \pi]$. We assume that there exists a fixed integer $M\in\mathbb{Z}_+$ such that the perturbation factor $b_i^m$ satisfies
\begin{equation*}
b_i^m = \left\{
\begin{aligned}
    &1, && m<0~\mathrm{or}~m>M,\\
    &1+\eta_i^m, && 0\leq m\leq M,
\end{aligned}
\right.
\end{equation*}
where $\eta_i^m>-1$ and $i=1,\ldots,N$ indexes the scatterers in the fundamental cell $Y$. This setting corresponds to a compact perturbation, in the sense that only finitely many periodic cells are perturbed. For $0\leq m\leq M$, we define $\mathfrak{H}_m := \mathrm{diag}(\eta_1^m,\ldots,\eta_N^m)$, where each entry $\eta_i^m$ is repeated six times along the diagonal. It follows that $\mathfrak{B}_m=\bm{I} + \mathfrak{H}_m$ for $0\leq m\leq M$.

\subsubsection{Single scatterer}

We first consider the case in which each periodic cell contains a single scatterer. Assume that the scatterer $D$ is a ball of radius $r_0$. A direct calculation shows that the $6\times 6$ matrix $\mathfrak M$ is given by
$$
\mathfrak{M} = \mathrm{diag}\left(\frac{4\pi r_0^3}{3},\frac{4\pi r_0^3}{3},\frac{4\pi r_0^3}{3},\frac{8\pi r_0^5}{15},\frac{8\pi r_0^5}{15},\frac{8\pi r_0^5}{15}\right).
$$

In the following numerical simulations, all monomers are taken to be identical balls with radius $r_0=0.1$. The physical parameters are set as $\lambda=\mu=1$, $\rho=1$, $\delta_1=10^{-4}$, and $\tau_1=2$ for all monomers. The left panel of Figure~\ref{fig:eigenmodedecaymonomerdimer} shows a localized mode in the monomer system with 100 monomers and the scatterer centered at the origin is perturbed. The perturbation strength of the central monomer is $\eta=-0.2$. The numerical computation yields one mode at $\omega=0.216527$. Strong localization around the perturbation is clearly observed.

\begin{figure}[htbp]
    \centering
   \includegraphics[width=\textwidth]{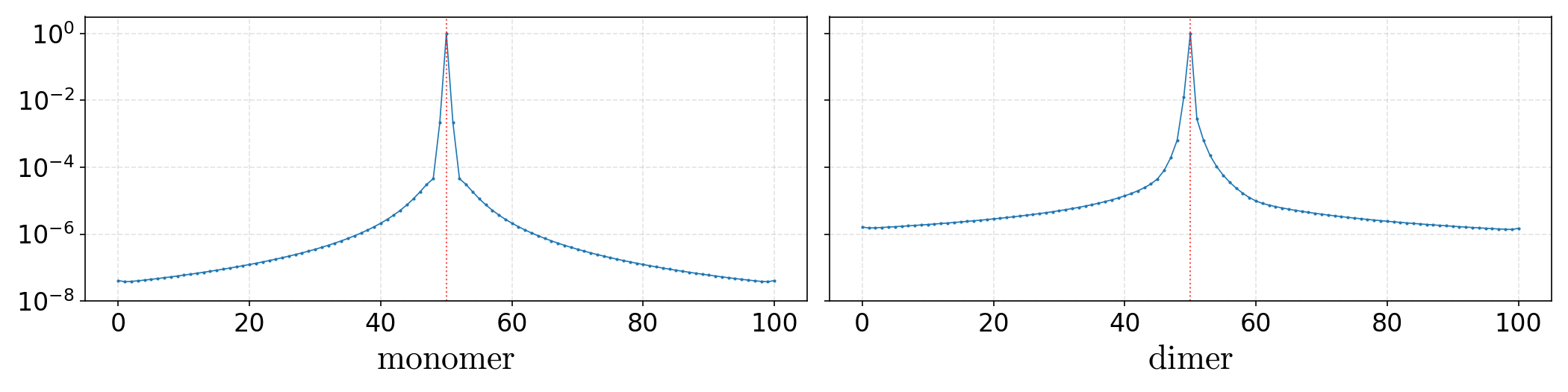}
    \caption{Decay of eigenmodes in monomer and dimer systems with a single perturbation.}
    \label{fig:eigenmodedecaymonomerdimer}
\end{figure}

Figure~\ref{fig:random50-1} shows several examples of a system with 50 monomers, randomly perturbed. We set the perturbation parameter as $\eta_1^m:= \eta(1+r)$, where $r$ is a random variable uniformly distributed in $(-1,1)$ and $\eta=-0.2$. Thus, $\eta_1^m$ lies in the interval $(-0.4,0)$. Among all eigenmodes, we select the one corresponding to the largest frequency. In each subplot, the upper panel shows the random perturbation $\eta_1^m$ applied to each monomer, while the lower panel displays the amplitude of the corresponding eigenmode. The plots show that, although the perturbation profiles differ significantly across the numerical examples, strong localization consistently emerges in each case.

\begin{figure}[htbp]
    \centering
    \includegraphics[width=\textwidth]{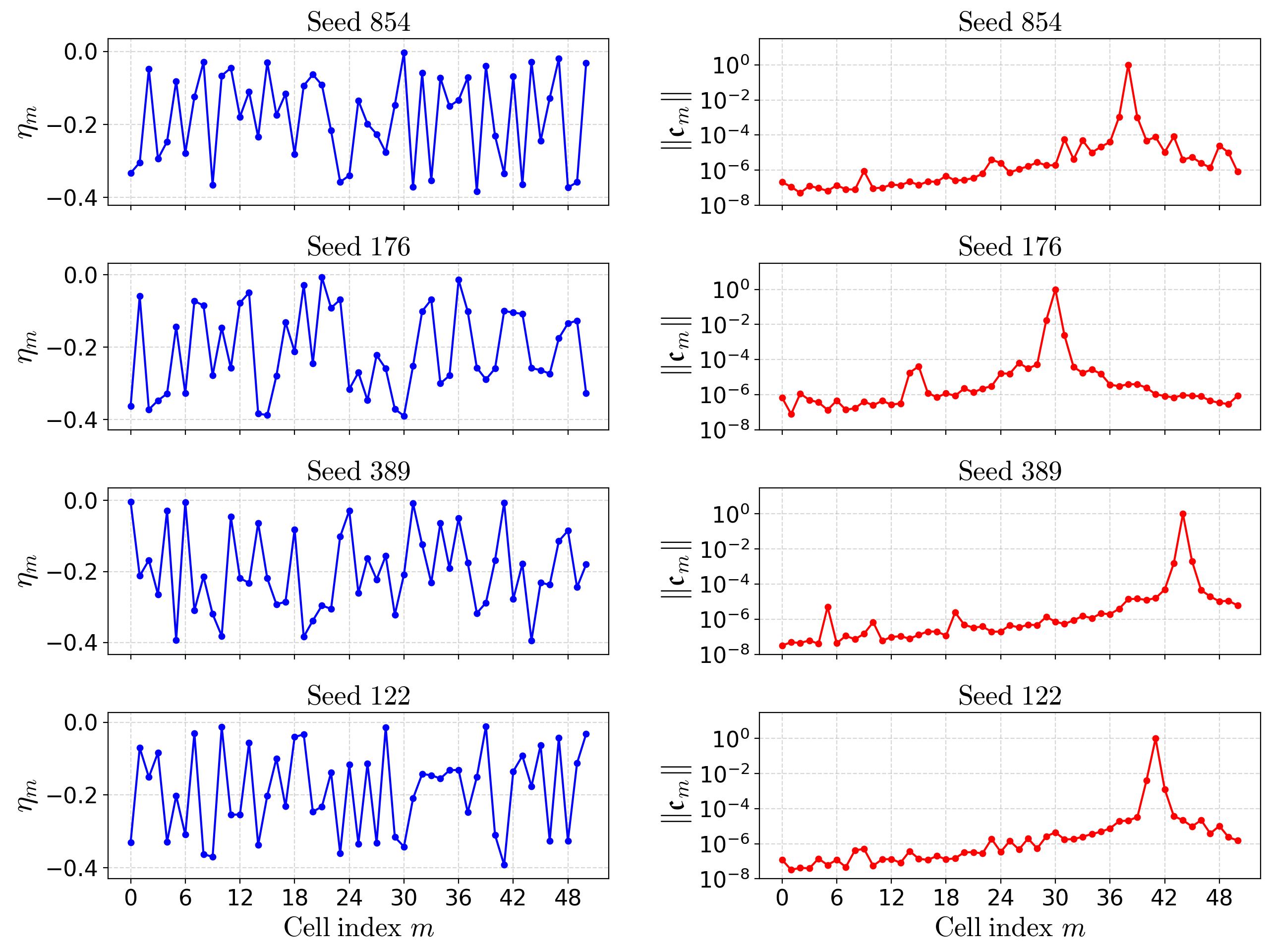}
     \caption{Random perturbations of 50 monomers.}
      \label{fig:random50-1}
\end{figure}

Figure~\ref{fig:eigenmodelocalization} illustrates the distributions of resonant frequencies and the average localization intensities of eigenmodes of 1000 randomly perturbed monomers for different frequency ranges and perturbation strengths. The localization intensity is defined as in (4.5) of \cite{Ammari2024} by
\begin{equation}\label{eq:localizationintensity}
    l(\bm{u}) = \frac{\|\mathfrak{c}_{\bm{m}}\|_{\infty}}{\|\mathfrak{c}_{\bm{m}}\|_2},
\end{equation}
where $\mathfrak{c}_{\bm{m}}$ denotes the solution of \eqref{eq:resfreqrel} associated with the eigenmode $\bm{u}$. As observed from the numerical results, the eigenmodes are separated into two distinct groups by a spectral gap. The lower-frequency group originates from the three lower bands among the lowest six energy bands, whereas the higher-frequency group corresponds to the three upper bands. The eigenmodes in the lower-frequency group exhibit significantly weaker localization compared with those in the higher-frequency group. As the perturbation strength increases, the frequency ranges of both groups expand, and the spectral gap is gradually filled by newly emerging localized modes. Moreover, the average localization intensity increases with the perturbation strength, indicating a progressive transition from propagating modes in the periodic system to localized modes in the presence of increasingly strong disorder. In the limiting regime with infinitely many random perturbations, this phenomenon corresponds to Anderson localization. The details of the numerical implementation are provided in Section~SM3.2 in supplementary material.

\begin{figure}[htbp]
    \centering
    \includegraphics[width=\textwidth]{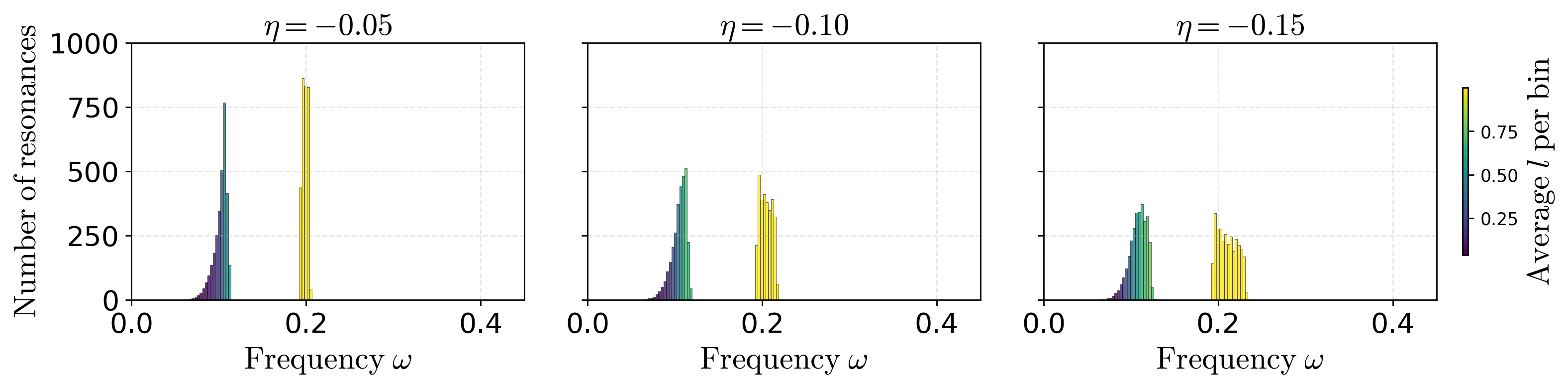}
    \includegraphics[width=\textwidth]{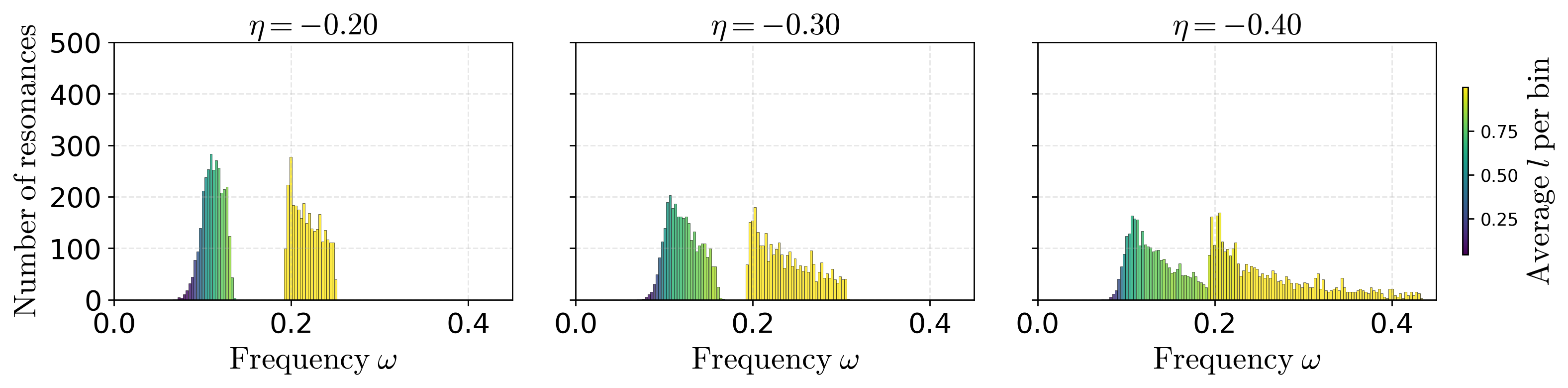}
    \caption{Distribution of eigenmodes in a randomly perturbed monomer system for different perturbation strengths.}
    \label{fig:eigenmodelocalization}
\end{figure}

\subsubsection{Two scatterers}

We now consider a slightly more complicated case in which each periodic cell contains two scatterers. Let $D:=D_1\cup D_2$ denote the union of the scatterers in the fundamental cell $Y=[-1/2,1/2]$. Assume that $D_1$ and $D_2$ are balls centered at $(-d,0,0)$ and $(d,0,0)$, respectively, with $d>0$. Both balls have radius $r_0$, where $0<r_0<d$. In this case, $\mathfrak M$ is a $12\times 12$ block diagonal matrix of the form
$$
    \mathfrak{M} = \begin{pmatrix}
    \mathfrak{M}_1 & 0 \\
    0 & \mathfrak{M}_2
    \end{pmatrix}.
$$
The $6\times 6$ blocks are given by $\mathfrak{M}_i:=\mathfrak{M}+\delta\mathfrak{M}_i, i=1,2$, where $\mathfrak M$ is the monomer mass matrix introduced in the previous subsection, and $\delta\mathfrak M_i=(\delta m_{jk,i})$ has the following nonzero entries:
\begin{small}
\begin{equation*}
\delta m_{jk,1}:=\left\{
\begin{aligned}
-\frac{4}{3}\pi r_0^3 d, ~ (j,k)=(2,4)~\mathrm{or}~(4,2),\\
\frac{4}{3}\pi r_0^3 d, ~ (j,k)=(3,6)~\mathrm{or}~(6,3), \\
\frac{4}{3}\pi r_0^3 d^2, ~ (j,k)=(4,4)~\mathrm{or}~(6,6),
\end{aligned}
\right.
\delta m_{jk,2}:=\left\{
\begin{aligned}
\frac{4}{3}\pi r_0^3 d, ~ (j,k)=(2,4)~\mathrm{or}~(4,2),\\
-\frac{4}{3}\pi r_0^3 d, ~(j,k)=(3,6)~\mathrm{or}~(6,3), \\
\frac{4}{3}\pi r_0^3 d^2, ~(j,k)=(4,4)~\mathrm{or}~(6,6).
\end{aligned}
\right.
\end{equation*}
\end{small}

All other entries of $\delta\mathfrak M_i$ are zero. The matrix $\widehat{\mathfrak{C}}_{\alpha}^\top$ is also of size $12\times 12$.

As in the previous subsection, we compute the resonant frequencies by identifying the values of $\omega$ for which \eqref{eq:resfreqrel} admits nontrivial solutions. In the numerical simulations for the dimer system, each fundamental cell contains two identical balls of radius $r_0=0.1$, centered at $(-0.2,0,0)$ and $(0.2,0,0)$, respectively. The physical parameters are $\lambda=\mu=1$, $\rho=1$, $\delta_1=10^{-4}$, and $\tau_1=2$ for all scatterers. The right panel of Figure~\ref{fig:eigenmodedecaymonomerdimer} shows a localized mode in the dimer system when the two scatterers in the cell centered at the origin are perturbed. The spherical scatterer centered at $(-0.2,0,0)$ is perturbed with strength $\eta=-0.2$, while the one centered at $(0.2,0,0)$ is perturbed with strength $\eta=-0.1$. The corresponding resonant frequency is $\omega=0.217165$.

Figure~\ref{fig:random50-2} shows several examples of a system with 50 randomly perturbed dimers. The perturbations applied to the two scatterers in each dimer are chosen independently at random. Specifically, we set $\eta_i^m:= \eta(1+r), i=1,2$, where $r$ is uniformly distributed in $(-1,1)$ and $\eta=-0.2$. Thus, $\eta_i^m$ lies in the interval $(-0.4,0)$. The blue lines represent the random perturbations and field amplitudes associated with the left scatterers in the dimers, while the red lines correspond to the right scatterers. The plots show strong localization for both scatterers within the same dimer, independently of the particular realization of the random perturbations.

\begin{figure}[htbp]
    \centering
    \includegraphics[width=\textwidth]{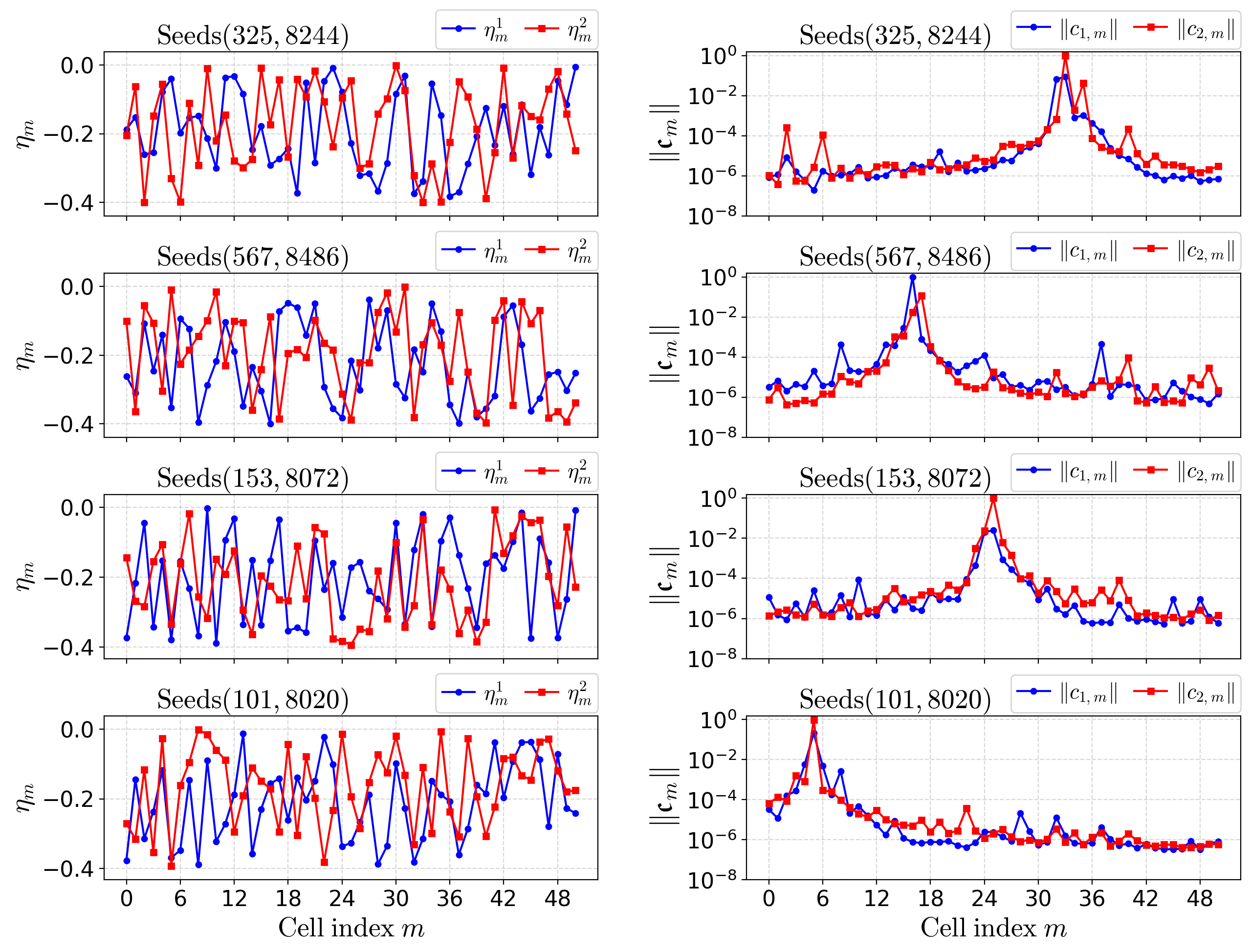}
     \caption{Random perturbations of 50 dimers.} 
     \label{fig:random50-2}
\end{figure}

Figure~\ref{fig:eigenmodelocalizationdimer} shows the distribution of resonant frequencies and the average localization intensity, defined in \eqref{eq:localizationintensity}, over different frequency intervals and for different perturbation strengths. 1000 dimers are randomly perturbed. A band gap between the lower- and higher-frequency modes is clearly observed. As the perturbation strength increases, the localization intensity also increases, consistent with the behavior observed in the monomer system. Details of the numerical implementation are provided in Section~SM3.3 in supplementary material.

\begin{figure}[htbp]
    \centering
    \includegraphics[width=\textwidth]{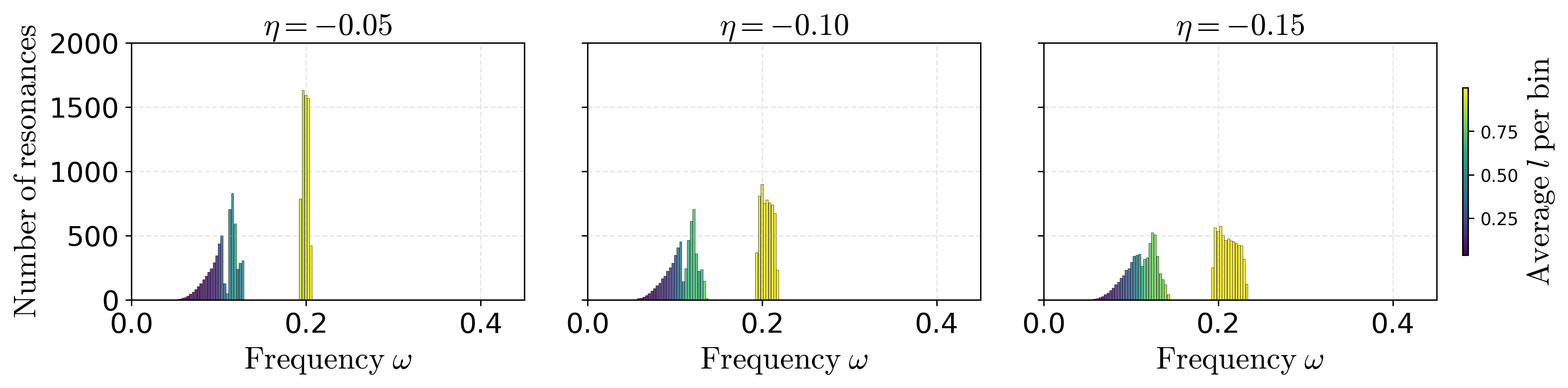}
   \includegraphics[width=\textwidth]{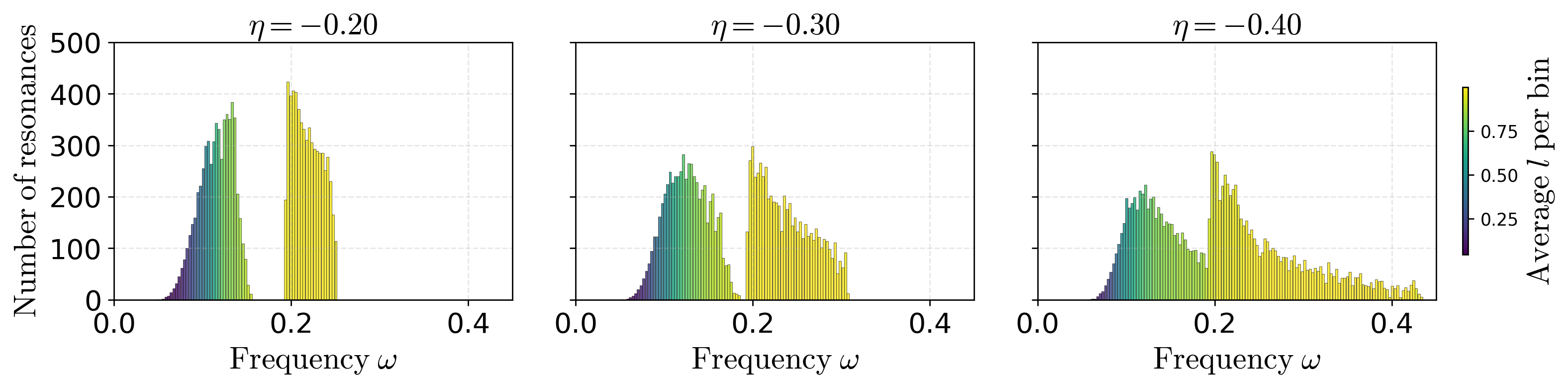}
    \caption{Distribution of eigenmodes in a randomly perturbed dimer system for different perturbation strengths.}
    \label{fig:eigenmodelocalizationdimer}
\end{figure}

\section{Conclusion}

In this paper, we present a mathematical analysis of subwavelength resonances in periodic elastic systems, both in the absence and in the presence of random perturbations. Using layer potential techniques, we reformulate the scattering problem as a system of boundary integral equations. Based on the properties of the associated layer potential operators and the generalized Rouch\'{e} theorem, we derive asymptotic approximations for the subwavelength resonant frequencies. For systems with random perturbations, we apply the Floquet transform to reformulate the problem in a periodic framework and derive the leading-order equations governing both the resonant frequencies and the total fields under general perturbations. Finally, we carry out numerical simulations for randomly perturbed periodic monomer and dimer systems. The numerical results agree well with the theoretical predictions and demonstrate the emergence of Anderson localization.

Due to the complexity of the problem, we did not pursue a full theoretical treatment of more general configurations. Anderson localization has been widely observed in various disordered systems, and understanding how different types of perturbations induce localization remains an important theoretical question. This provides a natural direction for future research.


\begin{thebibliography}{99}

\bibitem{AmmariBarandunCaoDaviesHiltunen2024SkinEffect}
H.~Ammari, S.~Barandun, J.~Cao, B.~Davies, and E.~O.~Hiltunen, Mathematical foundations of the non-Hermitian skin effect,
Arch. Ration. Mech. Anal., 248 (2024), no.~33.

\bibitem{ammari2015mathematical}
H.~Ammari, E.~Bretin, J.~Garnier, H.~Kang, H.~Lee, and A.~Wahab, Mathematical Methods in Elasticity Imaging, Princeton Ser. Appl. Math., Princeton Univ. Press, Princeton, NJ, 2015.

\bibitem{Ammari2024}
H.~Ammari, B.~Davies, and E.~O.~Hiltunen, Anderson localization in the subwavelength regime, Comm. Math. Phys., 405 (2024),
no.~1.

\bibitem{ammari2022exceptional}
H.~Ammari, B.~Davies, E.~O.~Hiltunen, H.~Lee, and S.~Yu, Exceptional points in parity--time-symmetric subwavelength
metamaterials, SIAM J. Math. Anal., 54 (2022), 6223--6253.

\bibitem{ammari2024functional}
H.~Ammari, B.~Davies, and E.~O.~Hiltunen, Functional analytic methods for discrete approximations of
subwavelength resonator systems, Pure Appl. Anal., 6 (2024), 873--939.

\bibitem{Ammari2018PhotonicsPhononics}
H.~Ammari, B.~Fitzpatrick, H.~Kang, M.~Ruiz, S.~Yu, and H.~Zhang, Mathematical and Computational Methods in Photonics and
Phononics, Math. Surveys Monogr., vol.~235, Amer. Math. Soc., Providence, RI, 2018.

\bibitem{AMMARI20175610}
H.~Ammari, B.~Fitzpatrick, H.~Lee, S.~Yu, and H.~Zhang, Subwavelength phononic bandgap opening in bubbly media, J. Differential Equations, 263 (2017), 5610--5629.

\bibitem{Ammari2007}
H.~Ammari and H.~Kang, Polarization and Moment Tensors: With Applications to Inverse Problems and Effective Medium Theory,
Appl. Math. Sci., vol.~162, Springer, New York, 2007.

\bibitem{AmmariKangLee2009}
H.~Ammari, H.~Kang, and H.~Lee, Layer Potential Techniques in Spectral Analysis, Math. Surveys Monogr., vol.~153, Amer. Math. Soc., Providence, RI, 2009.

\bibitem{AmmariKosche2024TopologicalHoneycombFloquet}
H.~Ammari and T.~Kosche, Topological phenomena in honeycomb Floquet metamaterial, Math. Ann., 388 (2024), 2755--2785.

\bibitem{PhysRev.109.1492}
P.~W.~Anderson, Absence of diffusion in certain random lattices, Phys. Rev., 109 (1958), 1492--1505.

\bibitem{BaoLi2022MaxwellPeriodic}
G.~Bao and P.~Li, Maxwell's Equations in Periodic Structures, Appl. Math. Sci., vol.~208, Springer, Singapore, 2022.

\bibitem{ChenGaoLiRen2025VariationalResonance}
B.~Chen, Y.~Gao, P.~Li, and Y.~Ren, Analysis of subwavelength resonances in high contrast elastic media
by a variational method, arXiv:2501.07315, 2025.

\bibitem{he2025uniquenessphaselessinverseelastic}
Y.~He, W.~Wu, and H.~Dang, Uniqueness for phaseless inverse elastic scattering problem for periodic structures in 2D and 3D,
Inverse Probl. Imaging, (2026), doi:10.3934/ipi.2026039.

\bibitem{John1987}
S.~John, Strong localization of photons in certain disordered dielectric superlattices, Phys. Rev. Lett., 58 (1987), 2486--2489.

\bibitem{doi:10.1137/S0036139997320536}
A.~Klein and A.~Figotin, Midgap defect modes in dielectric and acoustic media, SIAM J. Appl. Math., 58 (1998), 1748--1773.

\bibitem{LI2026113822}
H.~Li and L.~Xu, Resonant modes of two hard inclusions within a soft elastic material and their stress estimates,
J. Differential Equations, 453 (2026), no.~113822.

\bibitem{RenChenGaoLi2025SubwavelengthBandgaps}
Y.~Ren, B.~Chen, Y.~Gao, and P.~Li, Subwavelength phononic bandgaps in high-contrast elastic media, Multiscale Model. Simul.,  24 (2026), 429--455.

\bibitem{Sheng2006}
P.~Sheng, Introduction to Wave Scattering, Localization and Mesoscopic Phenomena, Springer Ser. Mater. Sci., vol.~88, 2nd ed., Springer, Berlin, 2006.

\bibitem{VodickaMantic2004}
R.~Vodi\v{c}ka and V.~Manti\v{c}, On invertibility of elastic single-layer potential operator, J. Elasticity, 74 (2004), 147--173.

\bibitem{Weaver1990}
R.~L.~Weaver, Anderson localization of ultrasound, Wave Motion, 12 (1990), 129--142.

\end{thebibliography}

\end{document}